\documentclass[a4paper,10pt]{article}
\usepackage[utf8]{inputenc}
\usepackage[margin=2cm]{geometry}
\usepackage[dvipsnames]{xcolor}
\usepackage{amsmath,amsthm,amssymb,bm,graphicx,bbm,enumerate,float}
\usepackage{natbib}
\usepackage{hyperref}
\hypersetup{colorlinks=true,linkcolor=OliveGreen,citecolor=MidnightBlue}

\newcommand{\E}{\mathrm{E}}
\renewcommand{\P}{\mathrm{P}}
\newcommand{\Q}{\mathrm{Q}}

\newtheorem{lemma}{Lemma}
\newtheorem{prop}{Proposition}
\newtheorem{cor}{Corollary}
\theoremstyle{definition}
\newtheorem{remark}{Remark}

\newcommand{\ba}{\bm{a}}
\newcommand{\bb}{\bm{b}}

\renewcommand{\hat}{\widehat}

\title{Higher-dimensional spatial extremes via single-site conditioning}
\author{J. L. Wadsworth and J. A. Tawn\\
	Department of Mathematics and Statistics, Fylde College, Lancaster University, LA1 4YF}
\date{}
\begin{document}

\maketitle

\begin{abstract}
Currently available models for spatial extremes suffer either from inflexibility in the dependence structures that they can capture, lack of scalability to high dimensions, or in most cases, both of these. We present an approach to spatial extreme value theory based on the conditional multivariate extreme value model, whereby the limit theory is formed through conditioning upon the value at a particular site being extreme. The ensuing methodology allows for a flexible class of dependence structures, as well as models that can be fitted in high dimensions. To overcome issues of conditioning on a single site, we suggest a joint inference scheme based on all observation locations, and implement an importance sampling algorithm to provide spatial realizations and estimates of quantities conditioning upon the process being extreme at any of one of an arbitrary set of locations. The modelling approach is applied to Australian summer temperature extremes, permitting assessment of the spatial extent of high temperature events over the continent.
\end{abstract}

\section{Introduction}
\label{sec:Intro}
\subsection{Background}
\label{sec:background}
In this work, our main motivation is to facilitate modelling of a continuous spatial process $\{Y(s): s \in \mathcal{S} \subset \mathbb{R}^2\}$ when it reaches extreme levels. We assume that $n$ replicates of the process $Y_i, i=1,\ldots,n$ are available in time, and that observations have been collected at a finite set of locations $s_1,\ldots, s_d$. To understand the risk posed by extreme events of $Y$, one firstly seeks general characterizations of spatial stochastic processes, given that some facet of them is extreme. These characterizations in turn yield statistical models that can be applied to appropriately selected extreme data, and used to investigate the probabilities of events that are more extreme than those observed to date.

Spatial extreme value theory has attracted much attention in recent years. A large tranche of literature deals with \emph{max-stable processes}, which arise through considering the limits of affinely normalized maxima $\{M_n(s) = \max_{1 \leq i \leq n} [Y_i(s) - b_n(s)]/a_n(s): s \in \mathcal{S}\}$ as the number of independent or weakly dependent copies $n$ tends to infinity. The functions $a_n(s)>0, b_n(s)$ are related only to the marginal distributions of $Y$; for details of this approach see e.g.\ \citet{Davisonetal12} or \citet[][Chapter 9]{deHaanFerreira07}. However, applications where the object of interest really is the spatial pointwise maximum, taken over a large number of repetitions, are relatively scarce. Moreover, inference for max-stable processes is notoriously difficult due to the fact that the pointwise maximum function $M_n$, and its limiting counterpart, is composed of a number of different underlying processes $Y_i$. This can be alleviated using the approach of \citet{StephensonTawn05} that includes information on which $Y_i$ contribute to $M_n$, but this can lead to bias if $d$ is relatively large compared to $n$ \citep{Wadsworth15}. Full likelihood inference ignoring this information, either via data augmentation \citep{Thibaudetal16} or Monte Carlo expectation maximization \citep{Huseretal19} is still difficult to implement and remains limited to moderate (up to $\approx 20$) numbers of locations.

The analog of max-stable processes for suitable definitions of threshold exceedances is Pareto processes \citep{FerreiradeHaan14}. This characterization focuses on the behaviour of $\{[Y(s) - b_n(s)]/a_n(s): \sup_{s \in \mathcal{S}} [Y(s) - b_n(s)]/a_n(s) >0\}$, with generalizations of this approach given by \citet{DombryRibatet15} and \citet{deFondevilleDavison20}. Likelihoods for the ensuing models are typically much easier to handle than for the corresponding max-stable process models, but the need for \emph{censored} likelihoods as a bias reduction technique \citep[e.g.][]{Huseretal16} still inhibits scalability. All realistic models for spatial extremes are based on suitably modified Gaussian processes, meaning that their censored likelihoods contain many evaluations of high-dimensional Gaussian cdfs, and standard implementations are limited to $\approx 30$ locations. For a certain popular class of Pareto processes, \citet{deFondevilleDavison18} used gradient score methods to avoid computing costly components of the likelihood, and combined this with coding efficiencies to permit inference on several hundred sites. Their application to a 3600-location problem is a notable exception to the small $d$ phenomenon observed elsewhere.

 In spite of such progress, when the goal of the analysis is extrapolation from observed levels further into the tail of the distribution, both max-stable and Pareto processes suffer from serious drawbacks in their applicability. Non-trivial limit processes arise as $n$ tends to infinity only when the underlying process $Y$ exhibits \emph{asymptotic dependence}, meaning that for marginal distributions $Y(s_j) \sim F_{s_j}$, the extremal dependence measure $\chi(s_1,s_2)>0$ for all $s_1 \neq s_2$, where $\chi(s_1,s_2) = \lim_{q \to 1} \chi_q(s_1,s_2)$ and
\begin{align}
 \chi_q(s_1,s_2) = \P\{F_{s_1}(Y(s_1))>q, F_{s_2}(Y(s_2))>q\}/(1-q) = \P\{Y(s_1)>F_{s_1}^{-1}(q) \mid Y(s_2)>F_{s_2}^{-1}(q)\}.\label{eq:chi}
\end{align}
When $\chi(s_1,s_2)=0$ instead, as is the case for all Gaussian processes that are not perfectly dependent, the limiting max-stable process is a collection of independent random variables, and the Pareto process cannot be sensibly defined, since it degenerates everywhere except where the supremum is realized. Testing the assumption of asymptotic dependence is critical for application of max-stable or Pareto models, since if it does not hold then their incorrect use will lead to bias in estimation of extreme event probabilities.  

A common feature of environmental data is that the dependence weakens as the level of the process increases: that is, for fixed $s_1,s_2$, $\chi_q(s_1,s_2) \searrow$ as $q \to 1$. Figure~\ref{fig:chi} displays estimates of $\{\chi_q(s_1,s_2): s_1 \neq s_2, q \in [0.95,0.995]\}$, plotted against distance on the $x$-axis and in different shades for different quantiles, for the Australian summer temperature data that will be analyzed in Section~\ref{sec:Australia}. Dependence clearly weakens with distance as expected. The smoothed lines give averages of the estimates over distance at the $q=0.95, 0.97$ and $0.99$ quantiles. The decrease in these lines indicates that dependence also weakens at higher quantiles, but this would not be captured by a fitted Pareto process (also shown), for which $\chi_q(s_1,s_2)>0$ does not depend on $q$.

\begin{figure}
\centering
 \includegraphics[width=0.5\textwidth]{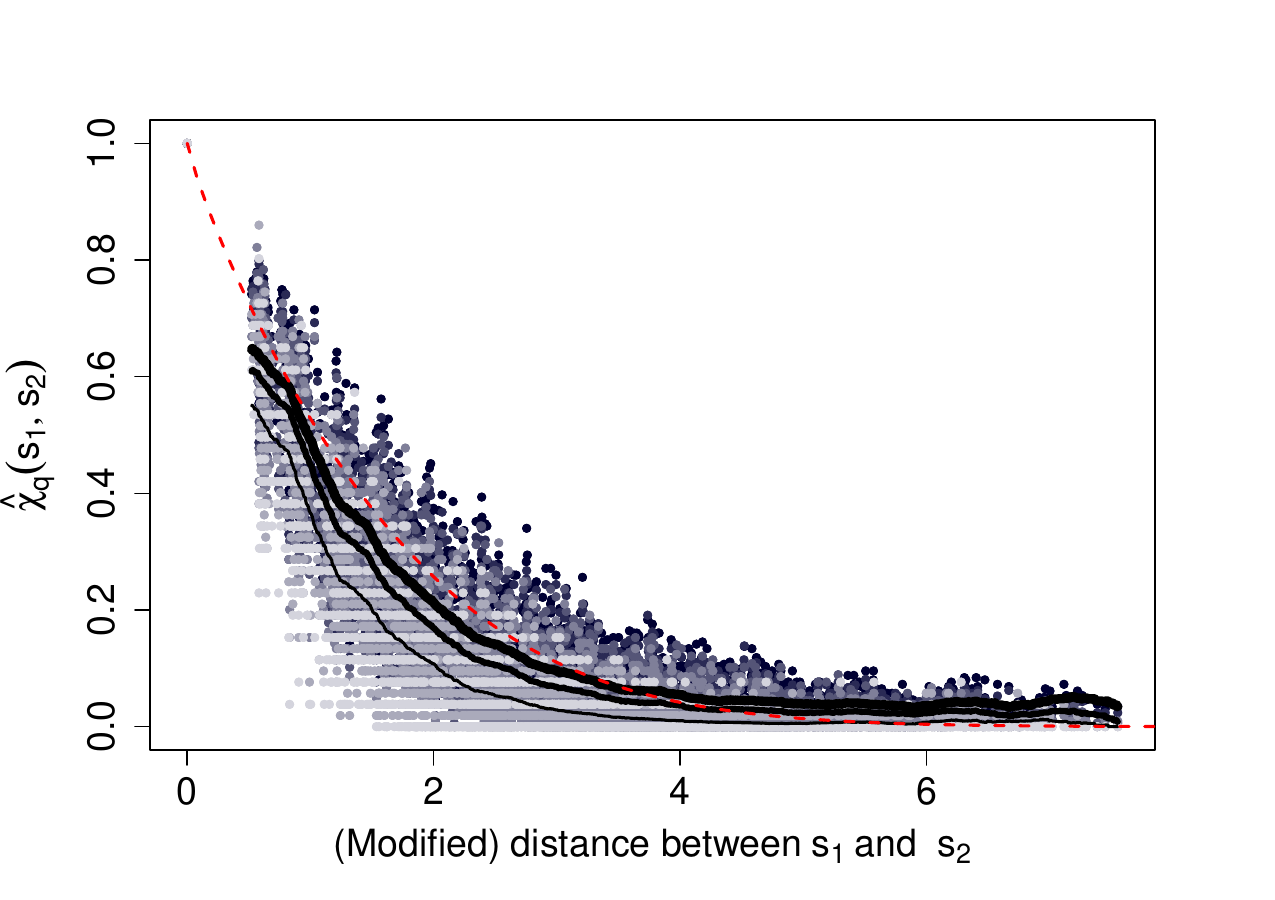}
 \caption{Estimates  $\widehat{\chi}_q(s_1,s_2)$ of $\chi_q(s_1,s_2)$ plotted against a modified version of distance $\|s_1-s_2\|$ (see Section~\ref{sec:Australia} for details of the spatial deformation). Darker points represent those estimated at lower quantiles, and lighter points at higher quantiles, moving from $q =0.95$ (darkest) to $q=0.995$ (lightest). The solid lines represent smoothed versions of $\widehat{\chi}_{0.95}(s_1,s_2)$ (thickest), $\widehat{\chi}_{0.97}(s_1,s_2)$ (medium) and $\widehat{\chi}_{0.99}(s_1,s_2)$ (thinnest). The dashed red line is the estimate of $\chi_q(s_1,s_2)$ from the Brown--Resnick Pareto process, which does not vary with $q$.}
 \label{fig:chi}
\end{figure}

In contrast to the asymptotic dependence case, little work has been done on developing asymptotically justifiable models for asymptotically independent extremes, for which $\chi_q(s_1,s_2) \searrow 0$. Partly, this is because one has to reconsider the meaning of ``asymptotically justifiable'' when the limits from classical extreme value theory are trivial. Based on a subasymptotic argument, \citet{WadsworthTawn12} suggested a class of models that might be broadly applicable to asymptotically independent extremes, whilst the Gaussian process forms another possibility \citep{Bortotetal00}. 

Often it is unclear whether data display asymptotic dependence or asymptotic independence, but models often only cover one dependence type. \citet{Huseretal17} and \citet{HuserWadsworth18} present spatial models that can capture both possibilities, although the same dependence class must hold over the entire spatial domain of interest. Furthermore, except for their Gaussian sub-models, independence between sites at long range is not possible. As such these models are suited only to smaller domains on which dependence persists. Very recent work \citep{Hazraetal21} extends the ideas of \citet{Huseretal17} to address this issue.

\subsection{Conditional extreme value theory: overview}
\label{sec:introCEVT}

We seek models for extremes of $Y$ that are asymptotically justified, appropriate for the features of environmental data, and that can be fitted to reasonably large numbers of observation locations. Our approach to this task is to exploit the so-called conditional extreme value model and adapt to the spatial setting. 

Conditional extreme value theory \citep{HeffernanTawn04,HeffernanResnick07} focuses on the behaviour of a random vector $\bm{Y}$, given that a fixed component of that random vector, say $Y_j$, is large. In contrast to classical multivariate extreme value theory, which leads to multivariate max-stable and Pareto distributions, resulting limit distributions can offer non-trivial descriptions of vectors exhibiting asymptotic dependence or asymptotic independence. We clarify that by ``conditional extreme value theory'' or the ``conditional model'', we mean theory or models derived explicitly through conditioning on a single component of a random vector, or value at a single site of a spatial process, being extreme. Other extreme value models (e.g.\ Pareto processes) are derived by conditioning on different extreme events, but the specific use of this terminology has become somewhat standard in the literature on extreme values \citep{DasResnick11b,ResnickZeber14}. In the time-series setting, there is related work that focuses on so-called tail chains, which describe the distribution of the process $\{Y_t: t \in \mathcal{T} \subset \mathbb{R}_+\}$ following an extreme event at $t=0$ \citep{BasrakSegers09, ResnickZeber14, Papastathopoulosetal17}. In the spatial case, we lose the directionality of time, and this situation has been much less explored. 

To take advantage of conditional extreme value theory in a spatial setting, we condition on the process $Y$ being extreme at an arbitrary reference location $s_0$. Based on an early formulation of the ideas in this paper, \citet{Tawnetal18} illustrated that the parameters of the \citet{HeffernanTawn04} model that describe pairwise dependence vary smoothly relative to distance between sites, giving an indication of the potential for extending that multivariate model to spatial processes when conditioning on a single location $s_0$. In practice, whilst this specific conditioning event might suit certain applications, there is often no natural reference location. More generally it is desirable to condition on the process being extreme anywhere over the domain of interest, $\mathcal{S}\subset\mathbb{R}^2$, as in the case of Pareto processes. We aim to combine the best of both worlds. Through the conditional extreme value approach we develop models that have flexible extremal dependence structures, accommodating asymptotic dependence and asymptotic independence. We then show how distributions conditioning on a single site can be combined to obtain inference on quantities of interest conditioning upon the maximum of the process at any collection of locations being large. This is achieved via an importance sampling algorithm, which also leads to an approximate method to simulate directly from the distribution of the field conditionally upon the maximum being large.

Aside from giving rise to flexible models for asymptotically (in)dependent data, the approach we propose can be scaled to accommodate several hundred observation locations. Whether asymptotically independent or dependent, most models for spatial extremes are driven by properties of the joint tail of $Y$, i.e., when all sites may be large together. Because of this, models are usually fitted via censored likelihood, so that locations with extreme observations contribute; see e.g.\ \citet{ThibaudOpitz15} or \citet{HuserWadsworth18} for examples. This censoring is a computational burden due to the integrals involved, although \citet{Zhangetal21} recently proposed the addition of a nugget effect to simplify likelihoods for simulation-based inference. The different asymptotic arguments for conditional extreme value theory do not necessitate censored likelihoods, making them inherently more scalable, although care can be required in ensuring that the implementation is sensible; further comments on this are given in Sections~\ref{sec:model} and~\ref{sec:Australia}.

In Section~\ref{sec:assumption} we detail the main assumption that will underlie the models built in Section~\ref{sec:model}, and which is illustrated through examples in Section~\ref{sec:examples}.

\subsection{Definitions and notation}
\label{sec:definitions}

Recall the definition of $\chi(s_1,s_2)$ given through equation~\eqref{eq:chi}. We specialize to the case of stationary processes $Y$, for which $\chi(s_1,s_2) \equiv \chi(h)$, $h=s_1-s_2$. A process is termed asymptotically dependent if $\chi(h)>0$ for all lags $h=(h_1,h_2)^\top$, and asymptotically independent if $\chi(h)=0$ for all $\|h\|>0$, with $\|\cdot\|$ the Euclidean distance. If $\chi(h)>0$ for $\|h\| < \Delta_\theta$, and $\chi(h)=0$ for $\|h\| \geq \Delta_\theta$, with $\theta=\arccos h_1/\|h\|$ the direction of $h$, we call the process \emph{directionally lag-asymptotically dependent}, or simply \emph{lag-asymptotically dependent} if the process is isotropic, so that $\Delta_\theta \equiv \Delta$.

Notationally, all vectors of length greater than one are expressed in boldface, with the exception of those denoting spatial location, e.g.\ $s_j=(s_{j,1},s_{j,2})^\top$, including spatial lags $h=s_1-s_2$. By convention, arithmetic operations on vectors are applied componentwise, with scalar values recycled as necessary. For example, if $f:\mathbb{R}\to\mathbb{R}$, $\bm{x} \in\mathbb{R}^d$, $c \in \mathbb{R}$, then $f(\bm{x}+c) = (f(x_1+c), \ldots, f(x_d+c))^\top$.

\section{Main assumption and examples}
\label{sec:assnexamples}

\subsection{Main assumption}
\label{sec:assumption}
The assumptions underlying conditional extreme value theory are most simply expressed following standardization to exponential-tailed margins, achieved in practice via the probability integral transformation. Let $\{X(s):s\in\mathcal{S}\subset\mathbb{R}^2\}$ be a stationary stochastic process with continuous sample paths and margins satisfying $\P(X(s_j)>x)\sim c e^{-x}$, $x\to\infty$, for some $c>0$. Where we wish to account for negative dependence, it is also assumed that $\P(X(s_j)<x) \sim c^* e^{-|x|}$, as $x\to -\infty$, $c^*>0$. For such cases, \citet{Keefetal13a} proposed modelling with Laplace margins. We assume that there exist functions $\{a_{s-s_0}: \mathbb{R} \to \mathbb{R}, s \in \mathcal{S}\}$, with $a_0(x)=x$ and $\{b_{s-s_0}: \mathbb{R} \to (0,\infty), s \in \mathcal{S}\}$, such that for any $s_0\in\mathcal{S}$, any $d \in \mathbb{N}$ and any collection of sites $s_1,\ldots,s_d\in\mathcal{S}$,  
 \begin{align}
   \left(\left\{\frac{X(s_i) - a_{s_i-s_0}(X(s_0))}{b_{s_i-s_0}(X(s_0))}\right\}_{i=1,\ldots,d},  X(s_0) - t \right) \big| X(s_0) > t  \overset{d}{\to} (\{Z^0(s_i)\}_{i=1,\ldots,d},E), \qquad t\to\infty; \label{eq:condlimit}
 \end{align}
that is, convergence in distribution of the normalized process to $\{Z^0(s): s \in \mathcal{S}\}$ in the sense of finite-dimensional distributions. In the limit, the variable $E \sim$ Exp(1) is independent of the residual process $Z^0$, which satisfies $Z^0(s_0) = 0$ almost surely, but is non-degenerate for all $s\neq s_0$ and places no mass at $+\infty$. Note that in assumption~\eqref{eq:condlimit} it is irrelevant whether the conditioning site $s_0$ represents one of the sites $\{s_1,\ldots,s_d\}$. When assumption~\eqref{eq:condlimit} is employed for modelling (see Sections~\ref{sec:model} and~\ref{sec:inference}) one indeed needs to condition on observed values. However, when simulating new events (see Section~\ref{sec:simulation}), the conditioning site could be any site in the domain.

An equivalent limiting formulation under the existence of a joint density for the process $X$ is obtained by conditioning upon the precise value of $X(s_0)$. This can be seen via l'H\^opital's rule:
\begin{align*}
\lim_{t\to\infty}  \frac{\P\left(\left\{\frac{X(s_i) - a_{s_i-s_0}(X(s_0))}{b_{s_i-s_0}(X(s_0))}\right\}_{i=1,\ldots,d}\leq \bm{z}, X(s_0)>t\right)}{\P(X(s_0)>t)} &= \lim_{t\to\infty}  \frac{ \frac{\partial}{\partial t}\P\left(\left\{\frac{X(s_i) - a_{s_i-s_0}(X(s_j))}{b_{s_i-s_0}(X(s_0))}\right\}_{i=1,\ldots,d}\leq \bm{z}, X(s_0)>t\right)}{\frac{\partial}{\partial t}\P(X(s_0)>t)}\\
& = \lim_{t\to\infty}  \P\left(\left\{\frac{X(s_i) - a_{s_i-s_0}(t)}{b_{s_i-s_0}(t)}\right\}_{i=1,\ldots,d}\leq \bm{z} \Big| X(s_0)=t\right).
\end{align*}
The independence of the conditioning variable is also assured under this alternative formulation of the assumption, see e.g.\ \citet[][Proposition~5]{Wadsworthetal17}. Since all the processes that we will consider have densities, this alternative version may sometimes be useful (see Section~\ref{sec:examples}). Although the assumption of a density is very common for spatial statistical modelling, it may be of theoretical interest to consider the situation where this is not the case. We do not attempt to treat this here.

The functions $a_{s-s_0},b_{s-s_0}$ appearing in limit~\eqref{eq:condlimit} can be characterized to some degree. \citet{HeffernanResnick07} detail various requirements in a bivariate setting under the assumption that the conditioning variable has a regularly varying tail, i.e., power law type decay, but with no marginal assumptions on the other variable. If the conditioning variable has an exponential tail, these requirements translate to the existence of functions $\psi^1_{s-s_0},\psi^2_{s-s_0}$ such that for all $c\in \mathbb{R}$ 
\begin{align}
\psi^1_{s-s_0}(c) &= \lim_{t \to \infty} b_{s-s_0}(t+c)/b_{s-s_0}(t) & \psi^2_{s-s_0}(c) = \lim_{t\to\infty}\{a_{s-s_0}(t+c)-a_{s-s_0}(t)\}/b_{s-s_0}(t), \label{eq:abcondns}
\end{align}
with local uniform convergence on compact subsets of $\mathbb{R}$; see also \citet{PapastathopoulosTawn16}. Conditions~\eqref{eq:abcondns} do not provide detailed information since the non-conditioning variable can have any marginal distribution, but any function $a_{s-s_0}$ or $b_{s-s_0}$ in~\eqref{eq:condlimit}, and later in Section~\ref{sec:model}, should satisfy these. Some modest additional structure is given in Proposition~\ref{prop:abcondns} of Appendix~\ref{app:proof} under assumptions on the support of the limit distribution. Generally the literature on conditional extremes is split according to whether the non-conditioning variable(s) are assumed to have a standard marginal form, such as the exponential tails in assumption~\eqref{eq:condlimit}.  Standardization occurs in most applied literature \citep[e.g.][]{Keefetal13a} with the prescription for the assumed normalization following from a variety of theoretical examples, in conjunction with checks on the modelling assumptions. We follow this general approach with the identical exponential-tailed margins in assumption~\eqref{eq:condlimit}. In more probabilistic literature \citep[e.g.][]{DasResnick11,DreesJanssen17} standardization is typically not assumed, or involves heavy tails rather than exponential tails, such as in \citet{HeffernanResnick07}.

A version of assumption~\eqref{eq:condlimit} suited only to processes exhibiting asymptotic dependence has been studied and applied in various places in the literature. In the spatial context, \citet{Engelkeetal14,Engelkeetal15} exploit~\eqref{eq:condlimit} with $a_{s-s_0}(x) \equiv x$, $b_{s-s_0}(x) \equiv 1$, i.e., the normalization does not depend on space. This is similar to the time series setting of \citet{BasrakSegers09}. The conditional extreme value approach has also been applied under asymptotic dependence in emerging applications such as graphical models and Markov trees \citep{EngelkeHitz20,Segers20}. Allowing a greater range of normalizations in~\eqref{eq:condlimit} is the key feature that permits treatment of a substantially broader set of underlying processes.

The more established approaches to spatial extreme value theory based on max-stable and Pareto processes can be justified for extremes of asymptotically dependent data not only through convergence of finite-dimensional distributions, but also functional regular variation \citep[e.g.][Chapter 9]{HultLindskog05,deFondevilleDavison20,deHaanFerreira07}. We note that our assumption in convergence~\eqref{eq:condlimit} is restricted to the case of finite-dimensional distributions. Functional convergence is not treated here, and although this would not alter the methodology that we introduce, we highlight that this is an area of potential further mathematical interest.

\subsection{Examples}
\label{sec:examples}

We present examples of widely used spatial processes that satisfy limit~\eqref{eq:condlimit}, and use these to identify useful structures for model building. Table~\ref{tab:ExampleSummary} summarizes the examples of Section~\ref{sec:Gaussian}--\ref{sec:HW}. 

\begin{table}[h]
\centering
\caption{Normalization functions and limit processes for theoretical examples given in Sections~\ref{sec:Gaussian}--\ref{sec:HW}. For the process of \citet{HuserWadsworth18}, $V$ in~\eqref{eq:HW} is taken as a marginally transformed Gaussian process. In the fourth row, $\ell_{s-s_0}(x)$ is a slowly-varying function of $x$, with $\lim_{x \to \infty}\ell_{s-s_0}(x) =0$, whose form is given by~\eqref{eq:IBRa}.}
 \begin{tabular}{|l|lll|}\hline
  Process & $a_{s-s_0}(x)$ & $b_{s-s_0}(x)$ & $Z^0(s)$ \\\hline
  Gaussian &$\rho(s-s_0)^2 x$ & $ 1+a_{s-s_0}(x)^{1/2}$& Gaussian \\
  Brown--Resnick &$x$ & $1$ & Gaussian \\
  $t_\nu$ & $x$ & $1$ & $t_{\nu+1}$ (transformed margins)\\
  Inverted Brown--Resnick & $\ell_{s-s_0}(x)x$ & $a_{s-s_0}(x)/(\log x)^{1/2}$ & Independence (reverse Gumbel margins) \\
  \citet{HuserWadsworth18} ($\lambda<1$) &$\rho(s-s_0)^2 x$ & $ 1+a_{s-s_0}(x)^{1/2}$& Gaussian \\\hline
 \end{tabular}
\label{tab:ExampleSummary}
\end{table}

\subsubsection{Gaussian process}
\label{sec:Gaussian}

 Let $\{Y(s):s\in\mathcal{S}\subset\mathbb{R}^2\}$ be a standard stationary Gaussian processes with correlation function $\rho(h)\geq 0$ and let $X(s) = -\log(1-\Phi(Y(s)))$ be the same process transformed to standard exponential margins. Then, taking $a_{s-s_0}(x) = \rho(s-s_0)^2 x$ and $b_{s-s_0}(x) = x^{1/2}$, the limit process $Z^0$ has finite-dimensional distributions that are Gaussian with zero mean and covariance structure defined by the matrix
\[
 \Sigma_0 = (2\rho_{k,0}\rho_{l,0} (\rho_{k,l}-\rho_{k,0}\rho_{l,0}))_{1\leq k,l\leq d},
\]
with $\rho_{k,0} = \rho(s_k-s_0)$ etc. See \citet{HeffernanTawn04} for more details of this derivation in a multivariate setting. This limit representation looks problematic as the process $Z^0$ becomes degenerate as $\rho(s_k-s_0) \to 0$, owing to the fact that the scale normalization, $X(s_0)^{1/2}$, is still present when $\rho(s_k-s_0)=0$ even though $X(s_k)$ and $X(s_0)$ are then independent. To avoid this, following \citet{WadsworthTawn13}, we can instead consider  
\[
  \left\{\frac{X(s_i) - \rho(s_i-s_0)^2X(s_0)}{\rho(s_i-s_0)X(s_0)^{1/2}+1}\right\}_{i=1,\ldots,d}, 
\]
i.e.\ taking $b_{s-s_0}(x) = 1+\rho(s-s_0)x^{1/2}$, for which the limit process $Z^0$ is Gaussian with zero mean and finite-dimensional covariance structure
 \begin{align*}
 \Sigma_0 = 2(\rho_{k,l}-\rho_{k,0}\rho_{l,0})_{1\leq k,l\leq d}. 
\end{align*}
Consequently, $Z^0$ has the conditional distribution of a Gaussian process with correlation function $\rho$, conditional upon the event $Z^0(s_0)=0$. However, when $\rho(s_k-s_0) = 0$, the limit process $Z^0$ does not have Gaussian margins, but identical margins to $X$. As such, there is still a discontinuity in the limit behaviour once independence is reached.

\subsubsection{Brown--Resnick process}
Let $\{X(s):s\in\mathcal{S}\subset\mathbb{R}^2\}$ be a Brown--Resnick process \citep{Kabluchkoetal09} with Gumbel margins. That is, $X$ can be expressed as
\begin{align}
 X(s) = \bigvee_{i=1}^\infty E_i + W_i(s) - \sigma^2(s)/2 \label{eq:BRc}
\end{align}
where $E_i$ are points of a Poisson process on $\mathbb{R}$ with intensity $e^{-x}dx$, $W_i$ are independent and identically distributed copies of a centred Gaussian process with stationary increments, and $\sigma^2(s) = \E\{W(s)^2\}$. The variogram of the process $W$ is given by
\begin{align}
\label{eq:vgm}
 \gamma(s_1,s_2) = \E[\{W(s_1)-W(s_2)\}^2],
\end{align}
and if $W(0)=0$ almost surely (a.s.), then $\sigma^2(s) = \gamma(s,0)$. Note that the representation~\eqref{eq:BRc} is not unique, and e.g., \citet{DiekerMikosh15} provide alternative representations with the same distribution.

\citet{Engelkeetal15} show that, for such a process, the finite-dimensional distributions of extremal increments $\{X(s)-X(s_0)|X(s_0)>t\}$ converges as $t\to\infty$ to a multivariate Gaussian distribution with mean vector and covariance matrix
\begin{align}
 \bm{m}_0 = (-\gamma(s_{i},s_0)/2)_{i=1,\ldots,d}, \qquad  \Omega_0 = (\gamma(s_{i},s_0)/2 + \gamma(s_{k},s_0)/2 - \gamma(s_i,s_k)/2)_{1\leq i, k \leq d}, \label{eq:M0Om0}
\end{align}
respectively. As such, the diagonal elements of the covariance matrix are $(\gamma(s_{i},s_0))_{i=1,\ldots,d}$, i.e., $\bm{m}_0=-\mbox{diag}(\Omega_0)/2$. This is the same limiting formulation as~\eqref{eq:condlimit}, with $a_{s-s_0}(x) = x$, $b_{s-s_0}(x) = 1$, and $Z^0$ a Gaussian process whose moment structure is determined by expression~\eqref{eq:M0Om0}.

\subsubsection{$t$ process}
The $t$ process, with $\nu>0$ degrees of freedom, arises as a particular Gaussian scale mixture. Specifically, taking
\begin{align}
 Y(s) = RW(s),
 \label{eq:tprocess}
\end{align}
with $W$ a standard Gaussian process with correlation function $\rho(h)$, and $R^{-2} \sim \mathrm{Gamma}(\nu/2,\nu/2)$, the finite-dimensional distributions $\bm{Y} = (Y(s_1),\ldots,Y(s_d))$ have density
\[
 f_\nu^d(\bm{y};\bm{\mu},\Sigma) = C_{\nu}^d [1+ (\bm{y}-\bm{\mu})^\top \Sigma^{-1} (\bm{y}-\bm{\mu})/\nu]^{-(\nu+d)/2},
\]
with normalization constant $C_{\nu}^{d} = \Gamma((\nu+d)/2) / \Gamma(\nu/2)|\Sigma|^{1/2}(\nu\pi)^{d/2}$; we write $\bm{Y} \sim \mathrm{St}_{\nu}^{d}(\bm{\mu},\Sigma)$. Calculations for the bivariate $t_\nu$ distribution were given in \citet{Keef}; here we extend these to arbitrary dimension.
Suppose $\bm{Y} \sim \mathrm{St}_{\nu}^{d}(\bm{0},\Sigma)$, with dispersion matrix $\Sigma = (\rho_{k,l})$ a correlation matrix, and let $\bm{Y}_{-j}$ represent $\bm{Y}$ without component $j$. Then
\begin{align}
 \frac{\bm{Y}_{-j} - \bm{\rho}_j{Y_j}}{(\nu + Y_j^2)^{1/2}} (\nu+1)^{1/2} \sim \mathrm{St}_{\nu+1}^{d-1}(\bm{0},(\rho_{k,l} - \rho_{j,k}\rho_{j,l})_{k,l \neq j}), \label{eq:tcond}
\end{align}
with $\bm{\rho}_j \in (-1,1)^{d-1}$ the $j$th column of $\Sigma$, without the $j$th component. Let $T:\mathbb{R}\to\mathbb{R}$ be a monotonic increasing transformation and consider $X = T(Y)$, such that $\P(X(s_j)>x)\sim e^{-x}$, with $\bm{X}=T(\bm{Y})$ the transformed finite-dimensional random vector. A suitable choice of transformation $T$ satisfies $T^{-1}(x) \sim K^{1/\nu}e^{x/\nu}$, $x \to \infty$, where $K=C_{\nu}^{1}\nu^{(\nu-1)/2}$. We then have
\begin{align*}
 \P(\bm{X}_{-j} - t \leq \bm{z} | X_j = t) &= \P(T(\bm{Y}_{-j}) \leq \bm{z}+t | T(Y_j) = t)\\
 & =  \P\left(\frac{\bm{Y}_{-j} - \bm{\rho}_j T^{-1}(t)}{(\nu + T^{-1}(t)^2)^{1/2}} \leq \frac{T^{-1}(\bm{z}+t) - \bm{\rho}_j T^{-1}(t)}{(\nu + T^{-1}(t)^2)^{1/2}} \Big| Y_j = T^{-1}(t)\right),
\end{align*}
which, by~\eqref{eq:tcond} is the cdf of the $\mathrm{St}_{\nu+1}^{d-1}(\bm{0},(\rho_{k,l} - \rho_{j,k}\rho_{j,l})_{k,l \neq j}/(\nu+1))$ distribution. Using the asymptotic form of $T^{-1}(x)$, the argument of this distribution function,
\begin{align*}
 \frac{T^{-1}(\bm{z}+t) - \bm{\rho}_j T^{-1}(t)}{(\nu + T^{-1}(t)^2)^{1/2}} \to e^{\bm{z}/\nu} - \bm{\rho}_j, \qquad t\to\infty.
\end{align*}
Consequently, $\bm{X}-X_j| X_{j}>t \overset{d}{\to} \bm{Z}^j$, where $Z_j^j = 0$ and
\[
 \P(\bm{Z}^j_{-j} \leq \bm{z}) = F_{\nu+1}^{d-1}(e^{\bm{z}/\nu}; \bm{\rho}_j, (\rho_{k,l} - \rho_{j,k}\rho_{j,l})_{k,l \neq j}/(\nu+1)),
\]
where $F_{\nu+1}^{d-1}(\cdot,\bm{\mu},\Sigma)$ is the cdf of $\mathrm{St}_{\nu+1}^{d-1}(\bm{\mu},\Sigma)$. As such, we conclude for the spatial process that $a_{s-s_0}(x) = x$, $b_{s-s_0}(x) = 1$. Note that this distribution has mass on lines through $-\infty$, which arise due to the fact that in representation~\eqref{eq:tprocess}, large values of $R$ cause both large and small values of $Y$, depending upon the sign of $W$.

The process $Z^0(s)$, defined through its finite-dimensional distributions $\bm{Z}^0=\{Z^0(s_1),\ldots,Z^0(s_d)\}$, is thus a transformed version of the $t_{\nu+1}$ process with $Z^0(s_0)=0$ a.s., with positive probability of being equal to $-\infty$ at any other site. That probability is smaller the stronger the dependence with the conditioning site; i.e., locally around the conditioning site, there is a high probability of $Z^0>-\infty$.

\subsubsection{Inverted Brown--Resnick process}
\citet{WadsworthTawn12} introduced the class of inverted max-stable processes as those processes whose upper joint tail has the same dependence structure as the lower joint tail of a max-stable process. That is, if $Y$ is a max-stable process and $T$ represents a monotonically-decreasing marginal transformation, then $T(Y)$ is an inverted max-stable process. Applying the transformation $T(x) = e^{-x}$ to~\eqref{eq:BRc} yields the inverted Brown--Resnick process, with standard exponential margins.

\citet{PapastathopoulosTawn16} consider conditional limits of the bivariate margins of the inverted Brown--Resnick process. They show that the normalizations required to obtain a non-degenerate limit are
\begin{align}
a_{s_i-s_0}(x) &= x  \exp\left\{\gamma(s_i,s_0)/4 -(\gamma(s_i,s_0)/2)^{1/2} (2\log x)^{1/2}+(\gamma(s_i,s_0)/2)^{1/2}\frac{\log\log x}{(2\log x)^{1/2}}\right\} \label{eq:IBRa}\\
b_{s_i-s_0}(x) &= a_{s_i-s_0}(x) / (\log x)^{1/2}, \label{eq:IBRb}
\end{align}
with limiting marginal distribution for $Z^0(s_i)$ given by 
\begin{align}
 \P(Z^0(s_i) \leq z) = 1-\exp\{-(\gamma(s_i,s_0)/2)^{1/2}\exp\{2z/\gamma(s_i,s_0)^{1/2})\}/(8\pi)^{1/2}\}. \label{eq:ZIBR}
\end{align}
The multiplier of $x$ in equation~\eqref{eq:IBRa} is a slowly varying function $\ell_{s-s_0}(x)$, meaning that for $c>0$, $\lim_{x \to \infty}\ell_{s-s_0}(cx)/\ell_{s-s_0}(x) = 1$. In the case of~\eqref{eq:IBRa}, $\lim_{x \to \infty}\ell_{s-s_0}(x) = 0$, i.e., $a_{s-s_0}(x) = o(x)$, $x \to \infty$. For the inverted Brown--Resnick process, the normalization and limiting distribution appear rather unintuitive when considering how $\gamma$ affects the dependence. Indeed, as $\gamma(s_i,s_0) \to \infty$, the max-stable process, and thus the inverted max-stable process, approaches independence. Yet, for finite $x$, the normalization $a_{s_i-s_0}(x)$ becomes large, and the mass of the limiting distribution of $Z^0(s_i)$ is placed at smaller values. Note that replacing $z$ by $\{\gamma(s_i-s_0)/2\}^{1/2} z -\gamma(s_i-s_0)^{1/2}\log\{\gamma(s_i-s_0)/2\} / 4$ removes dependence of the limit distribution on $\gamma$; this is equivalent to modifying the normalization functions to 
\begin{align*}
\tilde{a}_{s_i-s_0}(x) &= a_{s_i-s_0}(x) - b_{s_i-s_0}(x) \gamma(s_i,s_0)^{1/2}\log\{\gamma(s_i,s_0)/2\} / 4\\
\tilde{b}_{s_i-s_0}(x) &= \{\gamma(s_i-s_0)/2\}^{1/2} b_{s_i-s_0}(x).
\end{align*}
However, whilst stabilizing the limit in $\gamma$, this still does not lead to an easily interpretable normalization in the sense of $a$ and/or $b$ decreasing monotonically with $\gamma$. To understand this, it is helpful to consider how such a process is formed. From equation~\eqref{eq:BRc}, with $\sigma^2(s) = \gamma(s)$, the Brown--Resnick process is the pointwise maximum of location-adjusted Gaussian processes with negative drift. The more negative the drift (i.e., the larger $\gamma$) the more likely it is that the pointwise maxima from two locations will stem from different underlying Gaussian processes, which is why independence is achieved in the limit as $\gamma \to \infty$. Now, large values of the \emph{inverted} Brown--Resnick process correspond to small values of the original process, which are likely to be at the intersection points whereby different Gaussian processes contribute to the suprema. This rather complex construction thus leads to the seemingly unintuitive behaviour. Concerning the limiting process $Z^0$, \citet{PapastathopoulosTawn16} state that this corresponds to pointwise independence, and as such all structure lies in the functions $a_{s-s_0}$, $b_{s-s_0}$, and reverse Gumbel type margins~\eqref{eq:ZIBR}.

\subsubsection{Process of \citet{HuserWadsworth18}}
\label{sec:HW}

\citet{HuserWadsworth18} present a model for spatial extremes obtained by scale mixtures on Pareto margins, or location mixtures on exponential margins. Suppose that $V$ is an asymptotically independent spatial process with unit exponential margins, and $Q$ is an independent unit exponential variate. Then 
\begin{align}
 X^*(s) = \delta Q + (1-\delta) V(s) \label{eq:HW}
\end{align}
exhibits asymptotic independence for $\delta \leq 1/2$ and asymptotic dependence for $\delta>1/2$. Here, for simplicity of presentation, we reparameterize to $X(s) = X^{*}(s)/(1-\delta)$, and set $\lambda = \delta/(1-\delta) \in (0,\infty)$, with $\lambda \in (0,1]$ corresponding to asymptotic independence. The marginal distribution of $X$ is
\begin{align*}
 \P(X(s_j)>x) = \frac{1}{1-\lambda}e^{-x} - \frac{\lambda}{1-\lambda}e^{-x/\lambda},
\end{align*}
so that the leading order term is $e^{-x}/(1-\lambda)$ for $\lambda<1$ and $\lambda e^{-x/\lambda}/(\lambda-1)$ for $\lambda>1$; the case $\lambda=1$ is obtained as $(1+x)e^{-x}$ upon taking the limit. 

The following proposition establishes the behaviour of interest for the case $\lambda \in (0,1)$; in particular we find that the conditional limit distribution of the modified process $X$ requires the same normalization and has the same limit distribution as the process $V$, if the scale normalization required for $V$ is increasing in $V(s_0)$. For notational simplicity in the following we set $V_0=V(s_0)$, $X_0=X(s_0)$, $\bm{V} = (V(s_1),\ldots,V(s_d))$, $\bm{X} = (X(s_1),\ldots,X(s_d))$, $\bm{a}(x) = (a_{s_1-s_0}(x),\ldots,a_{s_d-s_0}(x))$, and $\bm{b}(x) = (b_{s_1-s_0}(x),\ldots,b_{s_d-s_0}(x))$.

\begin{prop}
\label{prop:sameab}
 Suppose that $\bm{V}$ has unit exponential margins, $\P(V_l>v) = e^{-v}, v >0$, and for $\ba(v)$ and $\bb(v)$ with twice-differentiable components $a_l, b_l$ satisfying $a_l'(v) \sim \alpha_l$, $a''_l(v) = o(1)$, $b_l'(v)/b_l(v) = o(1)$, as $v\to\infty$, $l=1,\ldots, d$,
 \begin{align*}
  \P\left(\frac{\bm{V}-\ba(V_0)}{\bb(V_0)} \leq \bm{z}~\Big|~V_0 = v\right) \to G(\bm{z}),\qquad v\to\infty. 
 \end{align*}
Suppose further that all first and second order partial derivatives with respect to components of $\bm{z}$ converge, as specified in Lemma~\ref{lem:convflexibility} of Appendix~\ref{app:proof}. Then for $\bm{X} = \bm{V} + \lambda Q$, with $Q \sim $Exp(1) independent of $\bm{V}$, and $\lambda \in(0,1)$,
 \[
\P\left(\frac{\bm{X}-\ba(X_0)}{\bb(X_0)} \leq \bm{z}~\Big|~X_0 = x\right) \to \begin{cases}
                                                                               G(\bm{z}), & \min_{1\leq l \leq d}b_l(x) \to \infty\\ 
                                                                               \int_0^{q^{\star}} G\left(\bm{z}+\frac{(\bm{\alpha}-1)\lambda q}{\lim_{x\to\infty} \bb(x)}\right))(1-\lambda)e^{-(1-\lambda)q}  \,\mathrm{d}q, & \mbox{otherwise} 
                                                                              \end{cases}
                                                                              \]
as $x\to\infty$, with $\bm{\alpha}=(\alpha_1,\ldots,\alpha_d)^\top$, and $q^{\star} = \lim_{x \to \infty} ([\min\{\min_l(a_l(x)+b_l(x) z_l),x\}]/\lambda)_+$.
 
\end{prop}

For $\lambda>1$, corresponding to $\delta>1/2$ in~\eqref{eq:HW}, asymptotic dependence arises, and with a rescaling, one can instead express $X(s) = Q + V(s)/\lambda$ such that $\P(X(s_j)>x) \sim (1-1/\lambda)^{-1}e^{-x}$. In this case, $a_{s-s_0}(x)=x$, $b_{s-s_0}(x)=1$ and the limiting dependence structure is determined by the distribution of $V$.

\begin{cor}
Assume the conditions of Proposition~\ref{prop:sameab}. The normalization and limit distribution for $X$ are the same as those for $V$ when either:
\begin{enumerate}
 \item[(i)] All $b_{s-s_0}(x) \to \infty$, $x\to\infty$, for all $s, s_0$.
 \item[(ii)]$a_{s-s_0}(x) \sim x$ and $b_{s-s_0}(x)\sim 1$, $x\to\infty$, for all $s,s_0$, as arises under asymptotic dependence for $V$.
\end{enumerate}
\end{cor}
The proof of Proposition~\ref{prop:sameab} is in Appendix~\ref{app:proof}. For the application of Proposition~\ref{prop:sameab}, convergence of the partial derivatives needs to be established. Supposing concretely that $V$ is a Gaussian process with margins transformed to be exponential, Lemma~\ref{lem:Gausspd} and Remark~\ref{rmk:Gauss2deriv} of Appendix~\ref{app:proof} provides this result.

The fact that the normalization and limit distribution are often the same for $X$ and $V$ may seem surprising given that \citet{HuserWadsworth18} noted differences in the so-called coefficient of tail dependence \citep{LedfordTawn96} for $X$ according to the relative values of $\delta$ in equation~\eqref{eq:HW} and the coefficient of tail dependence for $V$. This result highlights that the conditional extreme value approach provides a different asymptotic representation of the tail of a random vector or process to those where all variables are considered to be extreme simultaneously. Further elaboration of this can be found in \citet{NoldeWadsworth20}.

\section{Statistical modelling}
\label{sec:model}

Motivated by the limit assumption~\eqref{eq:condlimit} and the examples in Section~\ref{sec:examples}, we suppose that for a high threshold $u$,
\begin{align}
\label{eq:modelXj}
 \left\{X(s) | X(s_0) > u : s\in\mathcal{S}\right\} \overset{d}{\approx} \left\{a_{s-s_0}(X(s_0)) + b_{s-s_0}(X(s_0)) Z^{0}(s)  :s\in\mathcal{S}\right\},
\end{align}
for some choice of functions $a_{s-s_0}, b_{s-s_0}$, residual process $Z^0$, and that $X(s_0) - u | X(s_0) > u \sim$ Exp(1) is independent of $Z^0$. This specifies a process model conditioning on the value at a particular site $s_0$ being extreme. The process is observed at $d$ locations, so for inference we are interested in $d$ such specifications, taking each observation site as the conditioning site $s_0$. However, assuming stationarity in space, we obtain a common form for $a_{s-s_0}, b_{s-s_0}$ and $Z^0$ whatever the conditioning site. Inference is described in Section~\ref{sec:inference}, whilst simulation conditioning both on $\{X(s_0)>u\}$ and $\{\max_{i_1 \leq i \leq i_m}X(s_i) >u\}$, for any collection of sites $\{s_{i_1},\ldots, s_{i_m}\}$ with each $s_{i_j} \in \mathcal{S}$, is dealt with in Section~\ref{sec:simulation}. For now, we address the specification of $a_{s-s_0}$, $b_{s-s_0}$ and $Z^0$. Each of the proposed sets of functions $a_{s-s_0}, b_{s-s_0}$ satisfies conditions~\eqref{eq:abcondns}.

\subsection{Functions $a_{s-s_0}$ and $b_{s-s_0}$}
\label{sec:functionsab}
As noted in Section~\ref{sec:Intro}, the function $a_{s-s_0}:\mathbb{R}\to\mathbb{R}$ must satisfy $a_0(x) = x$. Further, if the process $X$ is asymptotically dependent, then $a_{s-s_0}(x) = x$ for all $s$, whilst if it is (directionally) lag-asymptotically dependent up to lag $\Delta_\theta$, then $a_{s-s_0}(x) = x$ for $\|s-s_0\|\leq\Delta_\theta$. Under asymptotic independence, and if $Z^0$ has any positive support, we have $a_{s-s_0}(x) < x$, where the bound may only hold asymptotically, i.e., as $x \to \infty$ (Proposition~\ref{prop:abcondns}, Appendix~\ref{app:proof}). Focusing on the isotropic case, we propose the general parametric form for $a_{s-s_0}$ as
\begin{align}
 a_{s-s_0}(x) = x\alpha(s-s_0)= \begin{cases}
                      x & \|s-s_0\| <\Delta\\
                      x \exp\{-[(\|s-s_0\|-\Delta)/\lambda]^\kappa\}, &  \|s-s_0\| \geq \Delta.
                     \end{cases}
\label{eq:alpha}
\end{align}
Taking $\Delta=0$ provides a flexible model for asymptotic independence, with parameters $(\lambda,\kappa)$ to estimate. Taking $\Delta>0$ but less than the maximum distance between any two sites in the domain would allow for lag-asymptotic dependence; the value of $\Delta$ itself could be estimated directly, although profiling over $\Delta$ on a grid may be preferable. Taking $\Delta$ larger than the maximum distance between any two sites would correspond to asymptotic dependence, in which case $(\lambda,\kappa)$ would not be estimated. Equation~\eqref{eq:alpha} covers all forms for $a_{s-s_0}$ from lines 1,2,3 and 5 of Table~\ref{tab:ExampleSummary}, if $\rho$ is a powered exponential correlation function, and line 4 up to a slowly varying function. Furthermore, condition~\eqref{eq:abcondns} is satisfied as long as $b_{s-s_0}(x) \not\to 0$ as $x \to \infty$. We note that other monotonically decreasing functions are also candidates for the second line of~\eqref{eq:alpha}; certain correlation functions and survival functions are natural choices. 

For the function $b_{s-s_0}$, we propose three forms to achieve different modelling aims.

\paragraph{Model 1} $b_{s-s_0}(x) = [1+ \zeta x^\beta]^{-1}$, $\beta<0$, $\zeta\geq 0$.

The rationale behind this suggestion is that, if $\zeta>0$, then $b_{s-s_0}(x) \nearrow 1$, with $\zeta$ and $\beta$ controlling the convergence to the constant value. When $\|s-s_0\| < \Delta$ and data are (lag-)asymptotically dependent, this permits the model to display some \emph{subasymptotic} dependence, in the sense that for $\|s_k-s_0\|<\Delta$, and $x_q$ the $q$-quantile of $X$,
\begin{align*}
 \chi_q(s_k, s_0) = \P(X(s_k)>x_q|X(s_0)>x_q) &= \P(X(s_0)+[1+ \zeta X(s_0)^\beta]^{-1} Z^{0}(s_k)>x_q | X(s_0)>x_q)\\
 &= e^{x_q}\int_{x_q}^\infty  \P(Z^{0}(s_k)>(x_q-v)/[1+ \zeta v^\beta]^{-1}) e^{-v}\, \mathrm{d}v\\
 &=\int_0^\infty  \P(Z^{0}(s_k)>-r/[1+ \zeta (r+x_q)^\beta]^{-1}) e^{-r}\, \mathrm{d}r \\
 &\searrow \int_0^\infty  \P(Z^{0}(s_k)>-r) e^{-r}\, \mathrm{d}r, \qquad q\to 1, x_q\to\infty.
\end{align*}
Here, the dependence measure $\chi_q(s_k, s_0)$ is as defined in equation~\eqref{eq:chi}. In practice, spatial data nearly always exhibit values of $\chi_q(s_k, s_0)$ that decrease with the level $q$, although Pareto process models for asymptotic dependence cannot capture this feature, leading to overestimation of dependence at extreme levels. Taking $b_{s-s_0}$ as a constant function here with $a_{s-s_0}(x) = x$ would also yield asymptotic dependence with $\chi_q$ not varying with $q$. Such a model could be implemented by choice if the estimate of $\beta\ll0$, or $\zeta \approx 0$, for example. 

\paragraph{Model 2} Let $\Delta=0$ in~\eqref{eq:alpha}, and $b_{s-s_0}(x) = x^\beta$.

This represents a modelling strategy close to that proposed by \citet{HeffernanTawn04}. If the marginal support of $Z^0$ includes $(0,\infty)$, we require $\beta<1$ (Proposition~\ref{prop:abcondns}, Appendix~\ref{app:proof}), whilst conditions~\eqref{eq:abcondns} imply $\beta\geq0$. We have
\[
  \chi_q(s_k,s_0) = \int_0^\infty \P[Z^{0}(s_k)>\{x_q-\alpha(s_k-s_0)(r+x_q)\}/(r+x_q)^\beta] e^{-r}\, \mathrm{d}r \searrow 0, \qquad q\to1, x_q\to\infty,
\]
so that the rate of convergence to zero is controlled by $\beta$, in conjunction with the value of $\alpha(s_k-s_0)$ and the distribution of $Z^0(s_k)$. Such a model may perform well on data that exhibit asymptotic independence, but positive dependence over the whole region of study. In particular, since $ b_{s-s_0}(x)$ is still growing even as $a_{s-s_0}(x) \to 0$ $(\|s-s_0\|\to\infty)$, independence cannot be achieved as $\|s-s_0\|\to\infty$. This observation motivates our final model.

\paragraph{Model 3}  Let $\Delta=0$ in~\eqref{eq:alpha}, and $b_{s-s_0}(x) =  1+ a_{s-s_0}(x)^\beta, \qquad \beta>0$.

With Model 3, if $a_{s-s_0}(x) \to 0$ for fixed $x$ as $\|s-s_0\|\to\infty$, then $b_{s-s_0}(x) \to 1$ as $\|s-s_0\|\to\infty$. Thus, for $s$ sufficiently far from $s_0$, we would have
\begin{align*}
 X(s) | X(s_0)> u \overset{d}{\approx} Z^{0}(s).
\end{align*}
This final observation indicates that, under Model 3, if $s$ is far enough from $s_0$ that $X(s)$ and $X(s_0)$ are independent, the margins of the process $Z^0$ should be the same as the process $X$.

Models 1--3 provide a wide variety of structures, but we note that other choices could be considered.

\subsection{Process $Z^0(s)$}
\label{sec:processZ}
The random field $Z^0$ must satisfy $Z^0(s_0) = 0$. Supposing initially that we begin with an arbitrary Gaussian process $Z_G$, there are two natural ways to derive a process from $Z_G$ with this property:
\begin{enumerate}[(i)]
 \item Set $Z^{0}$ to have the distribution of $Z_G|Z_G(s_0)=0$
 \item Set $Z^{0}$ to have the distribution of $Z_G(s)-Z_G(s_0)$,
\end{enumerate}
with the resulting processes both again Gaussian. The initial process $Z_G$ might be stationary, and specified by a correlation function $\rho$ and variance $\sigma^2$, or have stationary increments, specified by a variogram $\gamma$, as in~\eqref{eq:vgm}; addition of the drift term $-\gamma(s,s_0)/2$ to the latter yields the process specified by~\eqref{eq:M0Om0}.

Taking such a Gaussian process may be a natural and parsimonious choice in many situations. Indeed, allowing non-Gaussian dependence in $Z^0$ would lead to intractable models for high-dimensions, although we emphasize that the dependence of the modelled process $X$ is not Gaussian. However, the choice of marginal distribution for $Z^0$ can impact upon the model; in particular, as noted at the end of Section~\ref{sec:functionsab}, we may wish $Z^0$ to have the same margins as $X$ in certain places, e.g., when $\|s-s_0\|$ is large.

We thus propose the following distributional family, that includes both Gaussian and Laplace marginal distributions. We say that a random variable $Z$ has a \emph{delta-Laplace} distribution, with location parameter $\mu\in\mathbb{R}$ and scale parameter $\sigma>0$, if its density is
\begin{align}
 f(z) = \frac{\delta}{2\sigma\Gamma(1/\delta)}\exp\left\{-\left|\frac{z-\mu}{\sigma}\right|^{\delta}\right\}, \qquad \delta>0, \label{eq:deltaLaplace}
\end{align}
with $\Gamma(\cdot)$ the standard gamma function. When $\delta=1$ this is the Laplace distribution. If $X$ has Laplace margins, as suggested in \citet{Keefetal13a}, Model 3 can be completed by allowing the parameter $\delta$ to depend on $s-s_0$, such that $\delta(s-s_0) \to 1$ as $\|s-s_0\| \to \infty$. When $\delta = 2$, density~\eqref{eq:deltaLaplace} is that of the Gaussian distribution with mean $\mu$ and variance $\sigma^2/2$. In general, if $Z$ has distribution~\eqref{eq:deltaLaplace}, then $\E(Z) = \mu$ and $\mathrm{Var}(Z) = \{\Gamma(3/\delta)/\Gamma(1/\delta)\}\sigma^2$.

In practice, the location and scale structure implied by the processes (i) and (ii) defined on Gaussian margins are passed through to the delta-Laplace margins via the probability integral transform. In some situations it may be desirable to incorporate alternative mean structures by specifying a particular parametric form for $\mu(s-s_0)$ that is different from those implied by (i) and (ii); further discussion on this is made in Section~\ref{sec:Australia}.

\section{Inference}
\label{sec:inference}

\subsection{Likelihood}
The models described in Section~\ref{sec:model} are suitable given an extreme in a single location. However, we would like to combine these models to allow for extremes in any observed location. Assuming stationarity, this motivates the use of a composite likelihood to estimate parameters of $a_{s-s_0},b_{s-s_0}$ and $Z^0$, all of which do not depend on the particular conditioning site $s_0$, which is now taken to be one of the observation locations. Specifically, denote these parameters by $\bm{\theta}=(\bm{\theta}_a,\bm{\theta}_b,\bm{\theta}_Z)$, where some subsets of parameters may be scalar. The likelihood based on the $n_j$ points $(X_i(s_1),\ldots,X_i(s_d))= (x_1^{i},\ldots,x_d^{i}), i=1,\ldots, n_j$, for which $X_i(s_j)>u$ is denoted
\begin{align}
L_j(\bm{\theta}) = \prod_{i=1}^{n_j} f_{Z^j}(\{[x_k^i - a_{s_k-s_j}(x_j^i; \bm{\theta}_a)]/b_{s_k-s_j}(x^i_j; \bm{\theta}_b)\}_{k\in\{1,\ldots,d\}\setminus j}; \bm{\theta}_Z)\prod_{k\in\{1,\ldots,d\}\setminus j} b_{s_k-s_j}(x^i_j; \bm{\theta}_b)^{-1}, \label{eq:liksinglesite}
\end{align}
with $f_{Z^j}$ the density of $Z^0$ at the observation locations excluding $s_j$ when $s_0=s_j$. We estimate $\bm{\theta}$ from the composite likelihood over all $d$ observed sites $L(\bm{\theta}) = \prod_{j=1}^d L_j(\bm{\theta})$. Composite likelihoods have been widely used in inference for max-stable processes \citep{Padoanetal10}, where lower-dimensional --- typically pairwise --- densities are used due to being substantially simpler to compute. Our usage here is different, as the components of the composite likelihood are full joint densities. However, 
any realizations of $X$ that are larger than $u$ at multiple sites will be counted more than once in the likelihood, with different conditioning sites. The motivation for the composite likelihood is thus to get a single set of parameters that, on average, represent well each of the distributions conditioning on a single site being large. Parameter uncertainty can be assessed by bootstrap techniques; see Section~\ref{sec:Australia}. The first product in equation~\eqref{eq:liksinglesite} arises from an assumption of independent replications of $X|X(s_j)>u$. This will often not be satisfied in practice, but for weak-moderate temporal dependence parameter consistency is unaffected by such misspecification \citep[e.g.][]{ChandlerBate07}, while we can take advantage of a block bootstrap for assessing parameter uncertainty.

An advantage of the conditional approach over other models is that the likelihood does not require integrals as in the censored likelihoods commonly used for asymptotically dependent or asymptotically independent models \citep[e.g.,][]{WadsworthTawn14,ThibaudOpitz15,Huseretal17,HuserWadsworth18}. Such integrals are typically the limiting factor in the number of sites to which one can fit the model, often meaning a reasonable upper limit is 20--30 sites. Figure~\ref{fig:comptime} in Appendix~\ref{app:Comptime} displays example computation times for the evaluation of a single likelihood $L_j$, and the full composite likelihood for different numbers of observation locations $d$ and different numbers of average repetitions $n_j$, both for Model 3. Based on these findings, it is perfectly feasible to optimize a likelihood for a moderate number of repetitions at 100--200 observation locations with no special computing power.  In a recent review paper, \citet{HuserWadsworth22} suggest cutting the computational burden of the composite likelihood $L(\bm{\theta})$ by taking a product over a subset of the full set of possible conditioning sites.

\subsection{Model selection and diagnostics}
\label{sec:modelselection}
If all options for $a_{s-s_0}, b_{s-s_0}$ and $Z^0$ from Sections~\ref{sec:functionsab}--\ref{sec:processZ} are considered, the number of candidate models becomes large. One way to inform the structure of the model is to perform pairwise fits, for which $Z^0$ is a scalar random variable, and plot parameter estimates against distance. This will help to uncover if and how $\alpha(s-s_0)$ should decrease with $\|s-s_0\|$, as well as evolution of parameters of $Z^0$ with $\|s-s_0\|$. 

Although the computation time for the full composite likelihood is reasonable, it is desirable to have a faster way to compare several candidate models. To achieve this one could select a single site $j$ and compare goodness of fit based on AIC using likelihood~\eqref{eq:liksinglesite}. Should a clear preference for a set of modelling choices emerge over investigation of a few such conditioning sites, these can be carried forward to a full composite likelihood fit. 

To assess overall goodness of fit, we suggest (i) assessing whether the distribution of fitted residual processes $\{\hat{Z}^0(s) = [X(s)-\hat{a}_{s-s_0}(X(s_0))]/\hat{b}_{s-s_0}(X(s_0))\}$ is close to the assumed distribution, (ii) comparing parameter estimates from pairwise fits to global fits, and (iii) verifying that simulations from the model are consistent with the observed distribution when one or more sites is extreme. The assumption of independence of the conditioning variable and the process $Z^0$ can be tested by summarizing the fitted $\hat{Z}^0$ processes with low-dimensional variables such as the empirical mean and variance, and testing whether they are independent of the conditioning variable. If dependence is implicated then a higher threshold could be considered.

As we will show in Section~\ref{sec:simulation}, simulation from this model is simple. As a result, diagnostics based on the intended model purpose can be explored, by comparing the behaviour of the data to the fitted model by simulation.  

\section{Simulation}
\label{sec:simulation}

Inference on quantities of interest, such as the probabilities of different extreme events, can be made via simulation. Given the nature of the model and assumptions, simulation conditional on being extreme at a particular location $s_0$ is straightforward, and the algorithm is outlined in Section~\ref{sec:SimCond1site}. In many circumstances, it is desirable to condition instead on the process being extreme at some part of the domain, but not at a specific location. We address methods for this in Section~\ref{sec:SimCondAnySite}.

\subsection{Simulation given an extreme at a specified site}
\label{sec:SimCond1site}
Under assumption~\eqref{eq:condlimit} and model~\eqref{eq:modelXj}, simulation of $X$, given an extreme value above a threshold $v\geq u$ at a specific location $s_0$, proceeds as follows:

\theoremstyle{definition}
\newtheorem{alg}{Algorithm}

\begin{alg}
 \label{alg:sims0}
 ~~
\begin{enumerate}
 \item Generate $X(s_0)|X(s_0)>v \sim$ Exp(1)
 \item Independently of $X(s_0)$, generate $\{Z^0(s):s \in \mathcal{S}\}$ from the model as specified in Section~\ref{sec:processZ}.
 \item Set $\{X(s)|X(s_0)>v: s \in \mathcal{S}\} = \{a_{s-s_0}(X(s_0))+b_{s-s_0}(X(s_0))Z^0(s):s \in \mathcal{S}\}$.
\end{enumerate}
\end{alg}

In practice, Algorithm~\ref{alg:sims0} is used to simulate at a finite $m$-dimensional collection of sites $\{s_{i_1},\ldots,s_{i_m}\}$, where the observation sites $\{s_1,\ldots,s_d\}$ may or may not be included. In particular, the conditioning site $s_0$ need not correspond to an observation location. We denote the $m$-dimensional distribution of $X|X(s_0)>v$, obtained via the approximation \eqref{eq:modelXj}, by $\Q_0^{v}$. Further, denote by $M=\{i_1,\ldots, i_m\}$ the set of indices corresponding to the simulation sites, and by $D=\{1,\ldots,d\}$ the indices of the observation locations, where we may have $M=D$. Letting $\bm{X}_M = \{X(s_{i_1}),\ldots,X(s_{i_m})\}$, Algorithm~\ref{alg:sims0} can be used to estimate quantities of the form
\begin{align}
\label{eq:target0}
 \E\{g(\bm{X}_M)|X(s_0)>v\} = \E_{\Q_0^v}\{g(\bm{X}_M)\},
\end{align}
for some $g:\mathbb{R}^m \to \mathbb{R}^l$, $l\geq 1$. Equation~\eqref{eq:target0} is convenient if one is either interested in conditioning on a specific location, or if $g(\cdot)$ is an indicator function for an event that involves $X(s_0)>v'\geq v$, since then the unconditional probability of this event is easily estimated using the exponential distribution of $X(s_0)$. As an illustration, suppose that the simulation sites are the observation sites, let $\mathbbm{1}(A)$ be the indicator function for the occurrence of event $A$, and $g(\bm{X}_M) = \mathbbm{1}(X(s_1)>v,\ldots,X(s_d)>v)$. Then, taking $s_0$ as any $s_j$, $j=1,\ldots,d$ we have
\begin{align*}
 \P(X(s_1)>v,\ldots,X(s_d)>v) = \E_{\Q_0^v}\{\mathbbm{1}(X(s_1)>v,\ldots,X(s_d)>v)\} \times \P(X(s_0)>v).
\end{align*}

\subsection{Simulation given an extreme in the domain}
\label{sec:SimCondAnySite}

In place of conditioning on a single site being extreme, interest often lies in estimating quantities of the form
\begin{align}
 \E\left\{g(\bm{X}_M)|\max_{i\in M}{X(s_i)}>v\right\}, \label{eq:target}
\end{align}
where if $m=|M|$ is large and the locations suitably arranged, $\{\max_{i\in M}{X(s_i)}>v\} \approx \{\sup_{\mathcal{S}} X(s)>v\}$. 

We denote the distribution of $\bm{X}_M|\max_{i \in M} X(s_i)> v$ by $\P^v$, abbreviate $X(s_i) =X_i$, $i \in M$ or $i \in D$, and let $\Q^v : = \sum_{i \in M} \pi_i \Q_i^v $, where $\Q_i^v$ represents $\Q_0^v$ at $s_0=s_i$ and
\[\pi_i = \frac{\P(X_i>v)}{\sum_{i\in M} \P(X_i>v)},\]
which equals $1/m$ if the margins are identical. In other words $\Q^v$ is the mixture distribution formed by selecting each $\Q_i^v$ with probability $\pi_i$. The support of $\Q^v$ is $\{\bm{x}\in\mathbb{R}^m:\max_{i \in M} x_i>v\}$. We note that because we standardize the margins of $X$, all margins are equal and some of the following algorithms can be simplified. However, we give the more general formulation below.

We firstly describe an existing approach to the problem of estimating~\eqref{eq:target} based on rejection sampling, together with a number of limitations of this method. To overcome these, we introduce a new approach in Section~\ref{sec:importance} based on importance sampling.

\subsubsection{Rejection sampling}
\label{sec:rejection}

\citet{Keefetal13a} propose a method to simulate from the distribution of $\bm{X}_D|\max_{i \in D} X(s_i)>u$, i.e., where $v=u$ and $M=D=\{1,\ldots,d\}$, the set of observation locations. Note the partition 
\[\left\{\bm{x}\in\mathbb{R}^d:\max_{i \in D} x_i>u\right\} = \cup_{j=1}^d \left\{\bm{x}\in\mathbb{R}^d: x_j>u, x_j = \max_{i \in D} x_i\right\},\]
and that $\P^u = \sum_{j \in D} \tilde{\pi}^u_j \Q_{j,\mbox{max}}^u$, where $\Q_{j,\mbox{max}}^u$ is the distribution of $\bm{X}|\{X_j>u, X_j=\max_{i \in D} X_i\}$, and $\tilde{\pi}^u_j = \P(X_j = \max_{i \in D} X_i|\max_{i \in D} X_i>u)$. Their algorithm is as follows:

\begin{alg}
 \label{alg:rejection}
 ~~
\begin{enumerate}
 \item Estimate $\tilde{\pi}^u_j$ using empirical proportions
 \item To simulate from $\Q_{j,\mbox{max}}^u$, draw from $\Q_j^u$ and reject unless $X_j=\max_{i \in D} X_i$
 \item Simulate from $\P^u$ by drawing from $\Q_{j,\mbox{max}}^u$ with probability $\tilde{\pi}^u_j$
\end{enumerate}
\end{alg}
\noindent
 This rejection method leads directly to draws from (an estimate of) $\P^u$, but has substantial drawbacks:
\begin{enumerate}[(i)]
 \item One can only condition on being large at the observation locations, rather than any set of locations;
 \item For high dimensions, the number of rejections to get a single draw from $\Q_{j,\mbox{max}}^u$ may be very high;
 \item To simulate from $\P^v$, $v \gg u$, one needs to firstly simulate from $\P^u$ and use the resulting draws to estimate the probabilities $\tilde{\pi}^v_j$, before proceeding with the rejection sampling approach to simulate $\Q_{j,\mbox{max}}^v$;
 \item If the maximum never occurs at a particular site in the sample, it will not occur there in simulations either.
\end{enumerate}

The final point can be addressed by simulating from the model to estimate $\tilde{\pi}^u_j$, at the cost of adding complexity to the algorithm. We propose an approach to estimate quantities of the form~\eqref{eq:target} via importance sampling, that avoids all of these drawbacks.

\subsubsection{Importance sampling}
\label{sec:importance}

\newcommand{\I}{\mathbbm{1}}
\newcommand{\Rm}{R_{\rm{max}}^v}

Compared to $\P^v$, the distribution $\Q^v$ samples more frequently in regions where multiple variables are extreme. However, this frequency is tractable: in the region where exactly $k$ variables are larger than $v$, $\Q^v$ samples $k$ times more observations than $\P^v$, because there are $k$ distributions $\Q_j^v$ covering this area. 

To formalize this connection, let $\P$ be a probability measure on $\mathcal{B}(\mathbb{R}^m)$, the Borel sigma algebra of $\mathbb{R}^m$, for which $\P(X_j>v)>0$ for all $j \in M$. Letting $\Rm = \{\bm{x}\in\mathbb{R}^m:\max_{i \in M} x_i >v\}$, and $\E, \E_{\P^v}, \E_{\Q^v_j}$ denote expectation with respect to $\P, \P^v$ and $\Q^v_j$ respectively, we can express
\begin{align*}
 \E_{\P^v}\{g(\bm{X}_M)\} & = \frac{\E\{g(\bm{X}_M)\I(\bm{X}_M \in \Rm)\}}{\E\{\I(\bm{X}_M \in \Rm)\}}, &  \E_{\Q^v_j}\{g(\bm{X}_M)\} & = \frac{\E\{g(\bm{X}_M)\I(X_j>v)\}}{\E\{\I(X_j >v)\}}, 
\end{align*}
and $\E_{\Q^v}\{g(\bm{X})\} = \sum_{i \in M} \pi_i \E_{\Q^v_i}\{g(\bm{X})\}$. Further note that $\Rm = \cup_{K}R_K^v$, where
\[
 R_K^v = \{\bm{x} \in \mathbb{R}^m: x_k>v, \mbox{ for all } k\in K, \mbox{ and } x_l \leq v \mbox{ for all } l \in M\setminus K \},\qquad K \subseteq M,
\]
and the union is over all $K$ in the power set of $M$, minus the empty set,  with all $R_K^v$ disjoint.

\begin{prop}\label{prop:importance}
 For any random vector $\bm{X}_M \in\mathbb{R}^m$, for which $\P(X_i>v)>0$ for all $i \in M$, and any function $g:\mathbb{R}^m \to \mathbb{R}^{l}$
 \begin{align}
  \E_{\P^v}\{g(\bm{X}_M)\} = \frac{\sum_{i\in M} \P(X_i>v)}{\P(\max_{i\in M} X_i>v)} \sum_K \E_{\Q^v}\{g(\bm{X}_M)\I(\bm{X}_M \in R_K^v)/|K|\}  = \frac{\sum_K \E_{\Q^v}\{g(\bm{X}_M)\I(\bm{X}_M \in R_K^v)/|K|\} }{\sum_K \E_{\Q^v}\{\I(\bm{X}_M \in R_K^v)/|K|\} }.\label{eq:puexp}
 \end{align}
\end{prop}
\begin{proof}
We have 
\begin{align}
 \E_{\P^v}\{g(\bm{X}_M)\} = \frac{\E\{g(\bm{X}_M)\I(\bm{X}_M\in\Rm)\}}{\E\{\I(\bm{X}_M\in\Rm)\}} = \frac{\sum_K \E\{g(\bm{X}_M)\I(\bm{X}_M \in R_K^v)\}}{\P(\max_{i \in M}{X_i}>u)}, \label{eq:pufull}
\end{align}
whilst for any $j\in K$ we also have
\begin{align*}
 \E\{g(\bm{X}_M)\I(\bm{X}_M \in R_K^v)\} = \frac{\E\{g(\bm{X}_M)\I(\bm{X}_M \in R_K^v)\I(X_j>v)\}}{\E\{\I(X_j>v)\}}\E\{\I(X_j>v)\}.
\end{align*}
Averaging over each $j \in K$ we get
\begin{align}
  \E\{g(\bm{X}_M)\I(\bm{X}_M \in R_K^v)\} &= \frac{1}{|K|}\sum_{j \in K}\frac{\E\{g(\bm{X}_M)\I(\bm{X}_M \in R_K^v)\I(X_j>v)\}}{\E\{\I(X_j>v)\}}\E\{\I(X_j>v)\} \notag\\
  &= \frac{1}{|K|}\sum_{j \in M}\frac{\E\{g(\bm{X}_M)\I(\bm{X}_M \in R_K^v)\I(X_j>v)\}}{\E\{\I(X_j>v)\}}\pi_j \sum_{i\in M}\P(X_i>v), \label{eq:pubit}
\end{align}
the second line following since for $j \in M\setminus K$, $\E\{g(\bm{X}_M)\I(\bm{X}_M \in R_K^v)\I(X_j>v)\} = 0$, and using the definition of $\pi_j$. Putting~\eqref{eq:pufull} and~\eqref{eq:pubit} together
\begin{align*}
  \E_{\P^v}\{g(\bm{X}_M)\} &= \frac{\sum_{i\in M}\P(X_i>v)}{\P(\max_{i \in M}{X_i}>v)} \sum_K \frac{1}{|K|} \sum_{j \in M} \pi_j \frac{\E\{g(\bm{X}_M)\I(\bm{X}_M \in R_K^v)\I(X_j>v)\}}{\E\{\I(X_j>v)\}}\\
  &= \frac{\sum_{i\in M}\P(X_i>v)}{\P(\max_{i \in M}{X_i}>v)} \sum_K \E_{\Q^v}\{g(\bm{X}_M)\I(\bm{X}_M \in R_K^v)/|K|\},
\end{align*}
which gives the first equality in~\eqref{eq:puexp}. For the second equality, note that 
\begin{align*}
 1 = \E_{\P^v}\{\I(\bm{X}_M \in \Rm)\} = \frac{\sum_{i\in M}\P(X_i>v)}{\P(\max_{i \in M}{X_i}>v)}\sum_K \E_{\Q^v}\{\I(\bm{X}_M \in \Rm) \I(\bm{X}_M \in R_K^v)/|K|\},
\end{align*}
and dividing through by this gives the result.
\end{proof}

We note that there is nothing specific to our general assumptions or extreme-value theory in this proposition: it holds for any random vector and any $v$ whether $v$ is extreme or not. For our purposes, $v$ is a high marginal threshold and the distribution $\Q^v$ is obtained through asymptotic theory via $\Q_j^v$, which can be simulated from as outlined in Section~\ref{sec:SimCond1site}. 

Thanks to Proposition~\ref{prop:importance}, we can estimate $\E\{g(\bm{X}_M)|\max_{i \in M}{X_i}>v\} = \E_{\P^v}\{g(\bm{X}_M)\}$ as follows:

\begin{alg}
 \label{alg:importance}
 ~~
 \begin{enumerate}
  \item With probability $\pi_j=1/m$, draw $\bm{X}_M$ from $\Q^v_j$, $j \in M$
  \item Repeat step 1 $n$ times to get $n$ draws $\bm{X}_M^{1},\ldots, \bm{X}_M^{n}$ from $\Q^v$
  \item Estimate the expectation $\E_{\P^v}\{g(\bm{X}_M)\}$ by
  \begin{align}
 \frac{\sum_{i=1}^n g(\bm{X}_M^i)/|\{j \in M :X^i_j>v\}|}{\sum_{i=1}^n 1/|\{j \in M :X^i_j>v\}|},\qquad \bm{X}_M^i \sim \Q^v. \label{eq:importance}
\end{align}
  
 \end{enumerate}
\end{alg}
\noindent
We note that the natural estimate of $\E_{\P^v}\{g(\bm{X}_M)\}$ from equation~\eqref{eq:puexp} would look like
  \begin{align*}
 \frac{\sum_K\sum_{i=1}^n g(\bm{X}_M^i)\I(\bm{X}_M^i \in R_K)/|K|}{\sum_K \sum_{i=1}^n \I(\bm{X}_M^i \in R_K)/|K|},\qquad \bm{X}_M^i \sim \Q^v,
\end{align*}
 but since it is only the size of $K$ that matters we do not need to sum over all possible sets $K$, but simply divide by the number of variables exceeding the threshold. Furthermore, if it is desired to have realizations from the distribution conditioning upon the maximum being large, i.e., of $\bm{X}_M|\max_{i \in M} X_i>u$, this can be achieved approximately by using the importance weights in expression~\eqref{eq:importance}. That is, we sub-sample $n'<n$ realizations from the collection $\{\bm{X}_M^i\}_{i=1}^n$ with probabilities proportional to $\{1/|\{j:X^i_j>v\}|\}_{i=1}^n$. 

\subsubsection{Unconditional probabilities}
\label{sec:unconditional}

Conditioning upon the process $X$ being extreme anywhere in the domain creates a desirable definition of an extreme event that does not require reference to an arbitrary reference location $s_0$. However, should we wish to undo the conditioning, we require estimation of the probability $\P(\max_{i \in M}X(s_i)>v)$, in place of the simpler marginal probability $\P(X(s_0)>v)$, which has a known form after marginal transformation. Nonetheless, Algorithm~\ref{alg:importance} can be exploited to estimate $\P(\max_{i \in M}X(s_i)>v)$, by noting that for any $j \in M$
\begin{align*}
 \P\left(\max_{i \in M}X(s_i)>v\right) & = \frac{\P(X(s_j)>v)}{\P\left(X(s_j)>v | \max_{i \in M}X(s_i)>v\right)},
\end{align*}
for which the denominator can be estimated by Algorithm~\ref{alg:importance} by taking $g(\bm{X}_M) = \I(X_j>v)$.
 
\subsection{Conditional infill simulation of an existing event}
\newcommand{\obs}{\rm{obs}}

Assume that we have observations of the process $X$ at a collection of sites indexed by $D' \subseteq D$ with $|D'|=d'$. We may wish to simulate $X$, conditionally upon the values of the observed process when extreme for at least one site, at an alternative collection of sites $\{s_{k_1},\ldots,s_{k_l}\}$, where we let $L=\{k_1,\ldots,k_l\}$ with $L \cap D' = \emptyset$. If there are missing data, or for checking model fit, then a natural candidate is to take $L= D \setminus D'$, which leads to simulation at all observation locations, conditional upon a subset of them. In contrast to purely Gaussian models, conditional simulation is more challenging from models tailored to spatial extremes, although it is possible for models based on elliptical processes such as those described in \citet{Huseretal17}. An algorithm was proposed for max-stable processes by \citet{Dombryetal13}.

Consider a realization $X_i$, with $X_i(s_j)>u$ for all $j\in J_i\subseteq D$, and  $X_i(s_j)<u$ for all $j\not\in J_i$. For each $j\in J_i$, we have 
\begin{align*}
 Z^j_i(s) = \frac{X_i(s)-a_{s-s_j}(X_i(s_j))}{b_{s-s_j}(X_i(s_j))},
\end{align*}
which, following the discussion in Section~\ref{sec:processZ}, is modelled as a (marginally-transformed) Gaussian process. Simulation of $\{Z^j(s_k), k \in L \} | \{Z^j(s_i), i \in D'\}$ can thus be achieved exploiting conditional simulation from a Gaussian process. Simplifying notation slightly, let $(\bm{Z}_L,\bm{Z}_{D'})$ represent an $(l+d')$-dimensional random vector from the multivariate Gaussian with mean $\bm{\mu}$ and covariance matrix $\Sigma$, and transformed to have delta-Laplace margins with mean vector $\bm{\mu}^\delta$ and scale parameter vector $\bm{\sigma}^\delta$. Partition $\bm{\mu}=(\bm{\mu}_L,\bm{\mu}_{D'})$, with $\bm{\mu}_L \in\mathbb{R}^l$, $\bm{\mu}_{D'} \in\mathbb{R}^{d'}$, and similarly for other vectors, whilst $\Sigma_{L D'} \in \mathbb{R}^{l\times d'}$ etc., represent the partitioned $\Sigma$. Let $F_\delta, \Phi, f_\delta$ and $\phi$ represent the univariate distribution functions and densities, respectively, of the delta-Laplace and Gaussian distributions, where $F_\delta(\bm{z}_L;\bm{\mu}^{\delta}_L,\bm{\sigma}^\delta_L) = (F_\delta(z_{L,k_1};\mu^{\delta}_{L,k_1},\sigma^\delta_{L,k_1}),\ldots,F_\delta(z_{L,k_l};\mu^{\delta}_{L,k_l},\sigma^\delta_{L,k_l}))^\top$ etc., and denote by $\phi_l(\cdot;\bm{m},\Omega)$ the $l$-dimensional multivariate Gaussian density with mean $\bm{m}$ and covariance $\Omega$.

The conditional density of $(\bm{Z}_L|\bm{Z}_{D'}=\bm{z}_{D'})$, denoted $f_{\bm{Z}_L|\bm{Z}_{D'}}$, can be expressed
\begin{align}
 f_{\bm{Z}_L|\bm{Z}_{D'}}(\bm{z}_L|\bm{z}_{D'}) = \phi_{l} \left[\Phi^{-1}\left\{F_\delta(\bm{z}_L;\bm{\mu}^{\delta}_L,\bm{\sigma}^\delta_L); \bm{\mu}_L,\bm{\sigma}_L\right\}; \bm{\mu}_{L|D'}, \Sigma_{L|D'}\right] \prod_{k \in L}\frac{f_{\delta}(z_{L,k}; \mu^{\delta}_{k}, \sigma^{\delta}_{k})}{\phi(\Phi^{-1}(F_\delta(z_{L,k};\mu^{\delta}_{L,k},\sigma^\delta_{L,k}); \mu_{L,k},\sigma_{L,k}); \mu_{L,k},\sigma_{L,k})}, \label{eq:Zcond}
\end{align}
with
\begin{align*}
 \bm{\mu}_{L|D'} &= \bm{\mu}_L - \Sigma_{LD'}\Sigma_{D'D'}^{-1}(\bm{\mu}_{D'}-\Phi^{-1}\left\{F_\delta(\bm{z}_{D'};\bm{\mu}^{\delta}_{D'},\bm{\sigma}^\delta_{D'}); \bm{\mu}_{D'},\bm{\sigma}_{D'}\right\})\\
 \Sigma_{L|D'} & = \Sigma_{LL} -\Sigma_{LD'}\Sigma_{D'D'}^{-1}\Sigma_{D'L}.
\end{align*}
To simulate from density~\eqref{eq:Zcond}:
\begin{alg}
\label{alg:condsim}
~~
\begin{enumerate}
 \item Given $\bm{z}_{D'}$, $\bm{\mu}, \Sigma, \bm{\mu}_{\delta}, \bm{\sigma}_\delta$, calculate $\bm{\mu}_{L|D'}$ and $\Sigma_{L|D'}$
 \item Simulate $\bm{Y}_L$ from the $l-$dimensional Gaussian with mean $\bm{\mu}_{L|D'}$ and covariance $\Sigma_{L|D'}$
 \item Apply the transformation $\bm{Z}_L = F_{\delta}^{-1}\{\Phi(\bm{Y}_L; \bm{\mu}_L, \bm{\sigma}_L); \bm{\mu}^{\delta}_L, \bm{\sigma}^{\delta}_L\}$.
\end{enumerate}
 \end{alg}

 Note the final step of Algorithm~\ref{alg:condsim} is not a typical application of the probability integral transform, since the margins of $\bm{Y}_L$ do not have location and scale parameters $\bm{\mu}_L, \bm{\sigma}_L$, but rather $\bm{\mu}_{L|D'}, \bm{\sigma}_{L|D'}$ where the latter is the square root of the diagonal of $\Sigma_{L|D'}$. When $\delta=2$ and $\bm{\sigma}_\delta = \bm{\sigma}\{\Gamma(1/2)/\Gamma(3/2)\}^{1/2} = 2^{1/2}\bm{\sigma}$, then this is just standard conditional simulation from the Gaussian.

Once draws from $\{Z^j(s_{k}), k \in L\} | \{Z^j(s_i), i \in D'\}$ have been made, then the process $X_i(s)|X_i(s_j)>u$ is recovered by setting
\[
 X_i(s_{k})|\{X_i(s_j)>u\} = a_{s_k-s_j}(X_i(s_j)) + b_{s_k-s_j}(X_i(s_j)) Z^j(s_k), \qquad k\in L,
\]
and this can be done for each $j \in J_i$. That is, we are conditioning both on the values at observed sites, and on the process being extreme at a specific site $s_j \in J_i$. 

The conditional simulation using Algorithm~\ref{alg:condsim} is illustrated in Figure~\ref{fig:CS}, using a subset of the Australian temperature data analyzed in Section~\ref{sec:Australia}. For these data, which are gridded and complete, conditional simulation may serve to help understand the fit of the model, but for other datasets it could be useful to deal with missing observations or to simulate at unobserved sites. The intended use of the conditional simulation may dictate whether one of these sites is more interesting or whether inference should be combined across all sites in $J_i$.

\begin{figure}
\centering
 \includegraphics[width=0.4\textwidth]{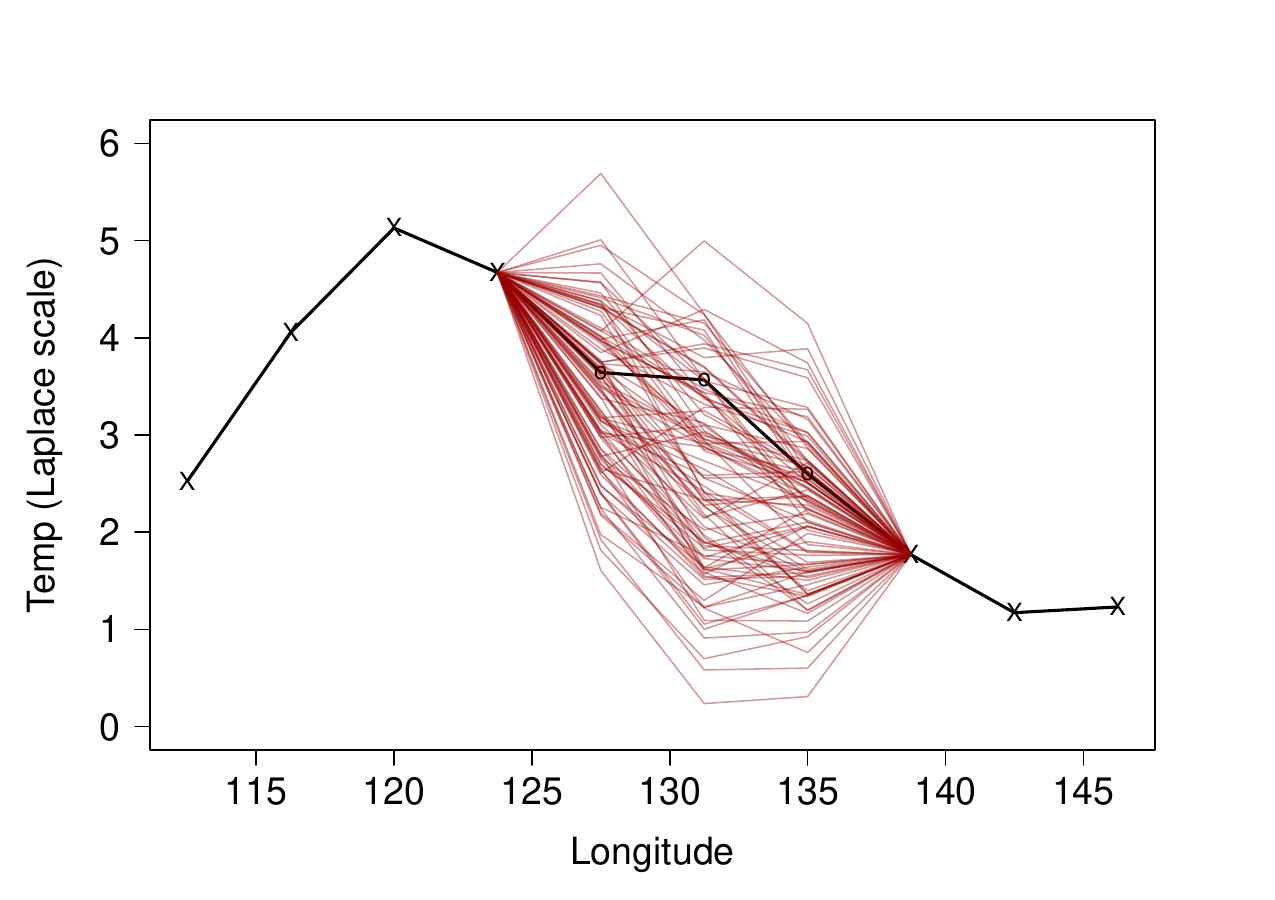}
  \caption{Example of conditional simulation. A realization of $X$ (black solid line) is displayed at 10 sites but the information from three sites (depicted by circles) is treated as missing. The red lines depict conditional simulations based on the other seven values, and being extreme at one of the sites above the modelling threshold. The model is the fitted model from Section~\ref{sec:Australia}, fitted to all observations.}
  \label{fig:CS}
\end{figure}

\section{Australian temperature extremes}
\label{sec:Australia}
\subsection{Data}
\label{sec:Data}
We analyze temperature data for Australia from the HadGHCND global gridded dataset \citep{Caesaretal06}. These are observed data translated onto a relatively coarse $2.5^\circ \times 3.75^\circ$ grid, with 72 points covering Australia; see the top-left panel of Figure~\ref{fig:deformation}. A complete record of daily maximum temperatures is available from 1/1/1957 - 31/12/2014, and to account for the seasonality of temperature data we focus on the summer temperatures recorded in December, January and February. This yields a total of 5234 days of observations at the 72 locations. The same data, recorded until the end of 2011, were analyzed in \citet{Winteretal16} using non-spatial models, where the focus lay in investigating the effects of the El Ni\~no Southern Oscillation (ENSO).

We transform the marginal distributions to unit Laplace, using the empirical distribution function at each site. In common with most works on extreme value dependence modelling, we do not attempt to explicitly account for the small amount of uncertainty introduced through use of estimated distribution functions. The choice of the double exponential-tailed Laplace distribution is motivated by the fact that we may anticipate independence at long range, and by allowing the margins of $Z^0$ to follow distribution~\eqref{eq:deltaLaplace}, this feature can be captured. An alternative to empirical transformation is to use a semi-parametric marginal transformation, modelling the upper tail with a univariate generalized Pareto distribution. In high-dimensional datasets this becomes more laborious as one needs to select a threshold and check the model at each site; failing to do so can lead to poor marginal fits that may impact upon dependence structure estimation. On the other hand, the semi-parametric transformation is more useful for back-transforming simulations to the original scale in the extremes.

\subsection{Exploratory analysis}
\label{sec:Exploratory}
Exploratory analysis indicates anisotropy and some spatial non-stationarity. This latter issue has received relatively little attention in an extreme value context, although \citet{HuserGenton16} consider non-stationary max-stable processes when covariate information is available. In more classical spatial modelling, one approach to dealing with non-stationarity is to deform the coordinate system, as suggested by \citet{SampsonGuttorp92}. The essential idea of this technique is to transform from one coordinate system (the ``G-plane'') in which non-stationarity is detected, to another configuration (the ``D-plane''), where stationarity might be more reasonably assumed. Such deformations have the advantage of being based on the data only, without need for covariates. We choose here to adopt such an approach, using the methodology of \citet{Smith96}.  Denote by $\tau:\mathbb{R}^2 \to \mathbb{R}^2$ the mapping from the G-plane to D-plane, with $s_j=(s_{j,1},s_{j,2})^\top \in \mathbb{R}^2$ and $\tau(s_j)=(\tau(s_j)_1,\tau(s_j)_2)^\top \in\mathbb{R}^2$. The technique of \citet{Smith96} is designed for estimation on all data and is driven by the covariance matrix: the task is to find a deformation function $\tau$ such that $\mathrm{Cov}\{Y(\tau(s_i)),Y(\tau(s_j))\}$ only depends on $\tau(s_i),\tau(s_j)$ through $\|\tau(s_i)-\tau(s_j)\|$. However, if patterns of non-stationarity are similar in the extremes to the body then this procedure should improve estimation in the extremes nonetheless. A similar argument was made by \citet{EastoeTawn09} in relation to marginal non-stationarity due to covariate effects. Assuming $d$ locations in the G-plane, \citet{Smith96} parameterizes 
\begin{align}
 \tau(s_j) = \left(\kappa^2 s_{j,1} + \psi \kappa\lambda s_{j,2} + \sum_{i=1}^d \omega_{1,i}\xi_i(s_j) , ~~\psi \kappa \lambda s_{j,1} + \lambda^2 s_{j,2} + \sum_{i=1}^d \omega_{2,i}\xi_i(s_j) \right)^\top, \qquad \kappa,\lambda>0, \psi\in\mathbb{R},\label{eq:tau}
\end{align}
with $\xi_i(s_j) = \frac{1}{2}\{(s_{j,1}-s_{i,1})^2 + (s_{j,2} - s_{i,2})^2\}\log\{(s_{j,1}-s_{i,1})^2 + (s_{j,2} - s_{i,2})^2\}$, and $\omega_{r,i}$, $r=1,2$, weights that satisfy $\sum_{i=1}^d \omega_{r,i} = \sum_{i=1}^d \omega_{r,i}s_{1,i}= \sum_{i=1}^d \omega_{r,i} s_{2,i} =0$. In practice, following \citet{Smith96}, we set most $\omega_{r,i}$ to zero, and focus on using a smaller number of geographically dispersed \emph{anchor sites} to estimate the deformation. Figure~\ref{fig:deformation} displays the G-plane and D-plane, highlighting the sites used in the estimation. Also displayed are estimates of $\chi_q(s_i,s_j)$ and $\chi_q(\tau(s_i),\tau(s_j))$ for $q=0.95$ against distance in each coordinate system. The more stationary the data, the less variability one expects to observe in $\chi_q$ for a given spatial lag $\|s_i-s_j\|$ or $\|\tau(s_i)-\tau(s_j)\|$: a modest improvement seems to arise using the D-plane over the G-plane. A number of factors account for this: reduced anisotropy, accounting for differences in distance between $1^\circ$ of latitude and longitude, as well as non-stationarity itself. Slightly different D-planes would be estimated using different anchor sites, and there is no clear way to optimize this aspect given the combinatorial possibilities. We proceed using the D-plane coordinates displayed in Figure~\ref{fig:deformation}. 

\begin{figure}[htb]
\centering
 \includegraphics[width=0.4\textwidth]{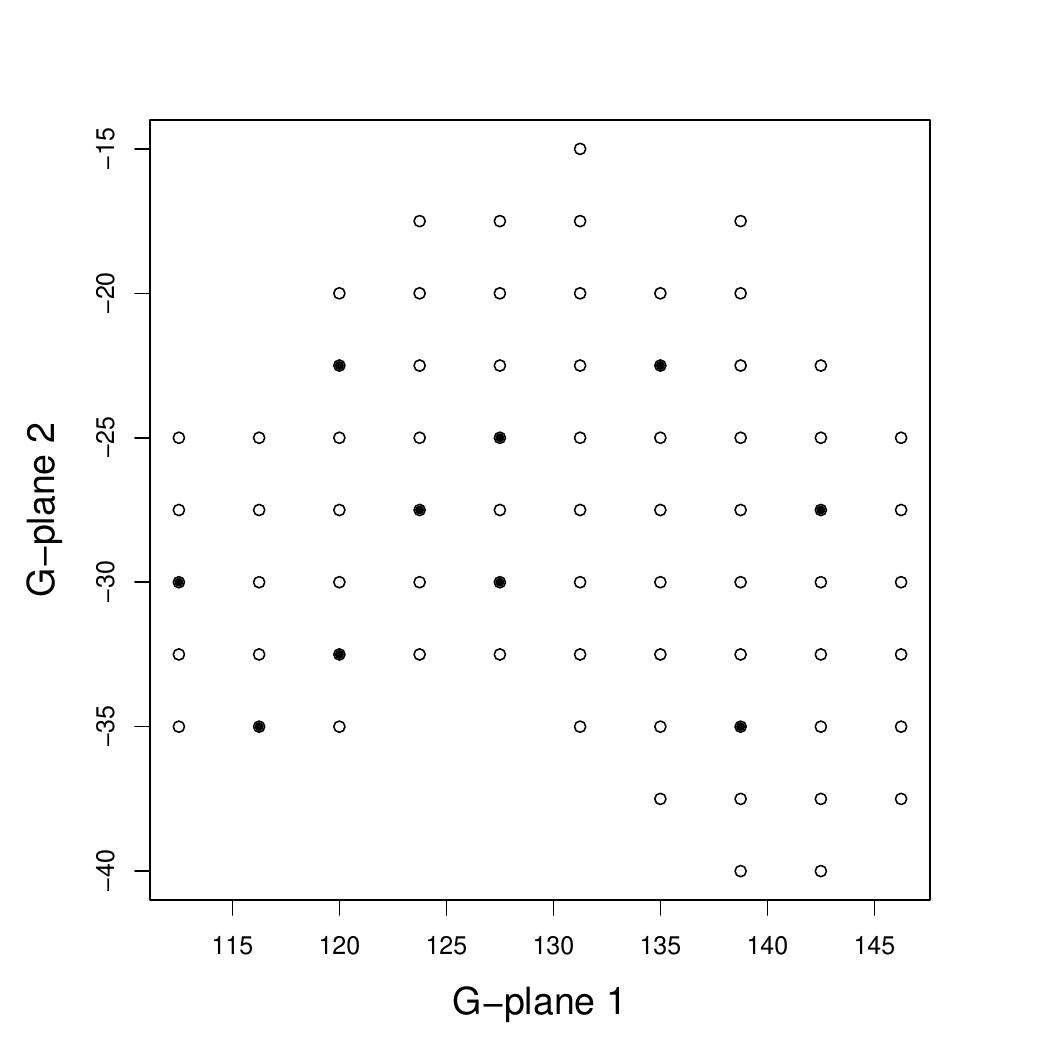}
 \includegraphics[width=0.4\textwidth]{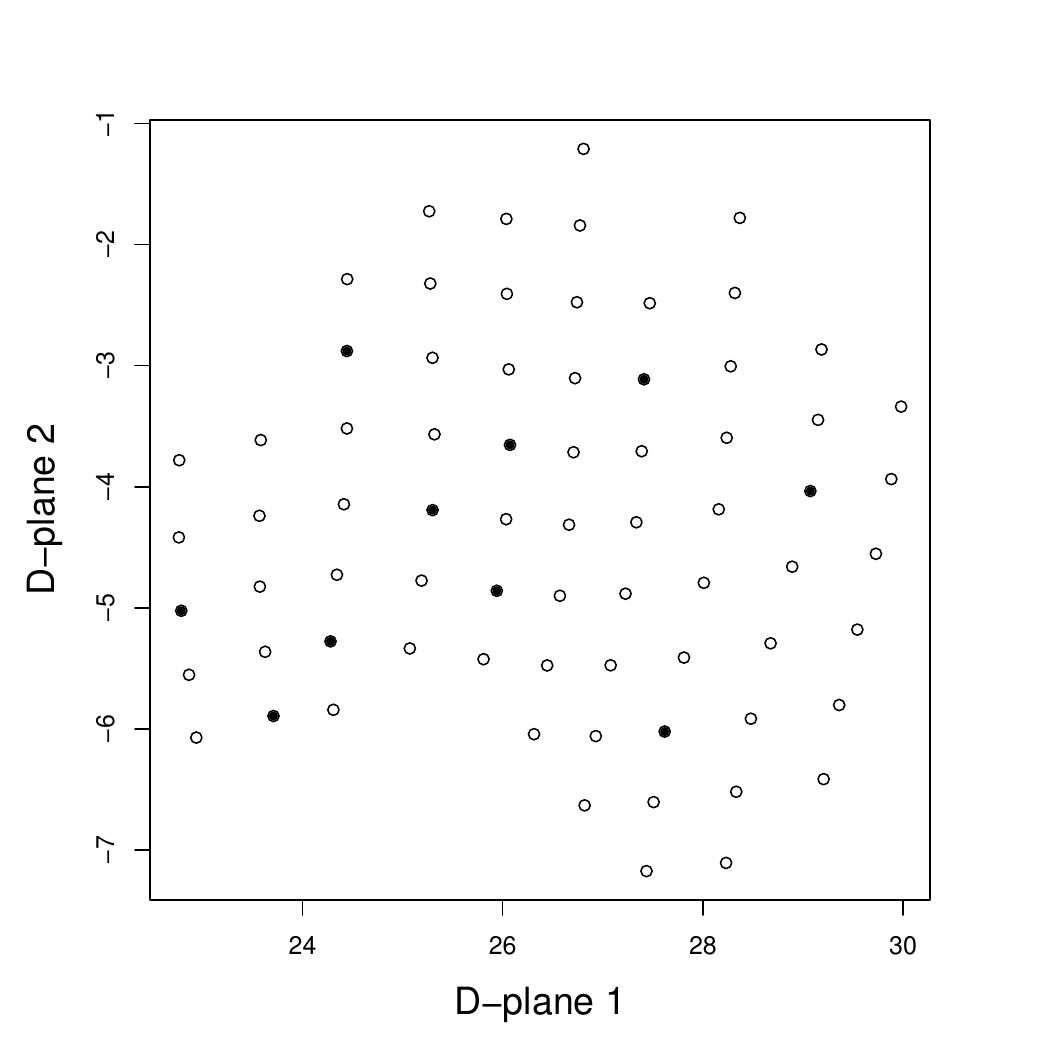}\\
 \includegraphics[width=0.4\textwidth]{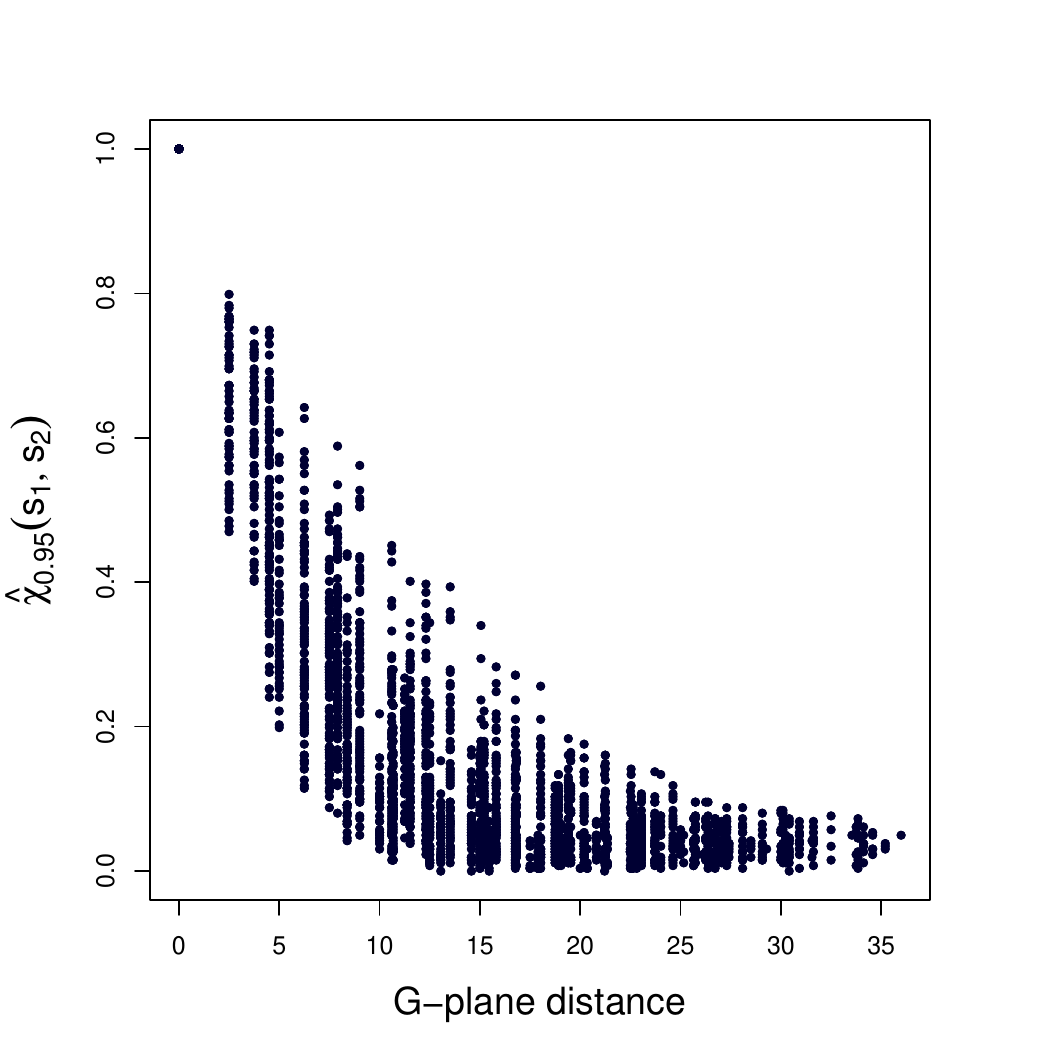}
 \includegraphics[width=0.4\textwidth]{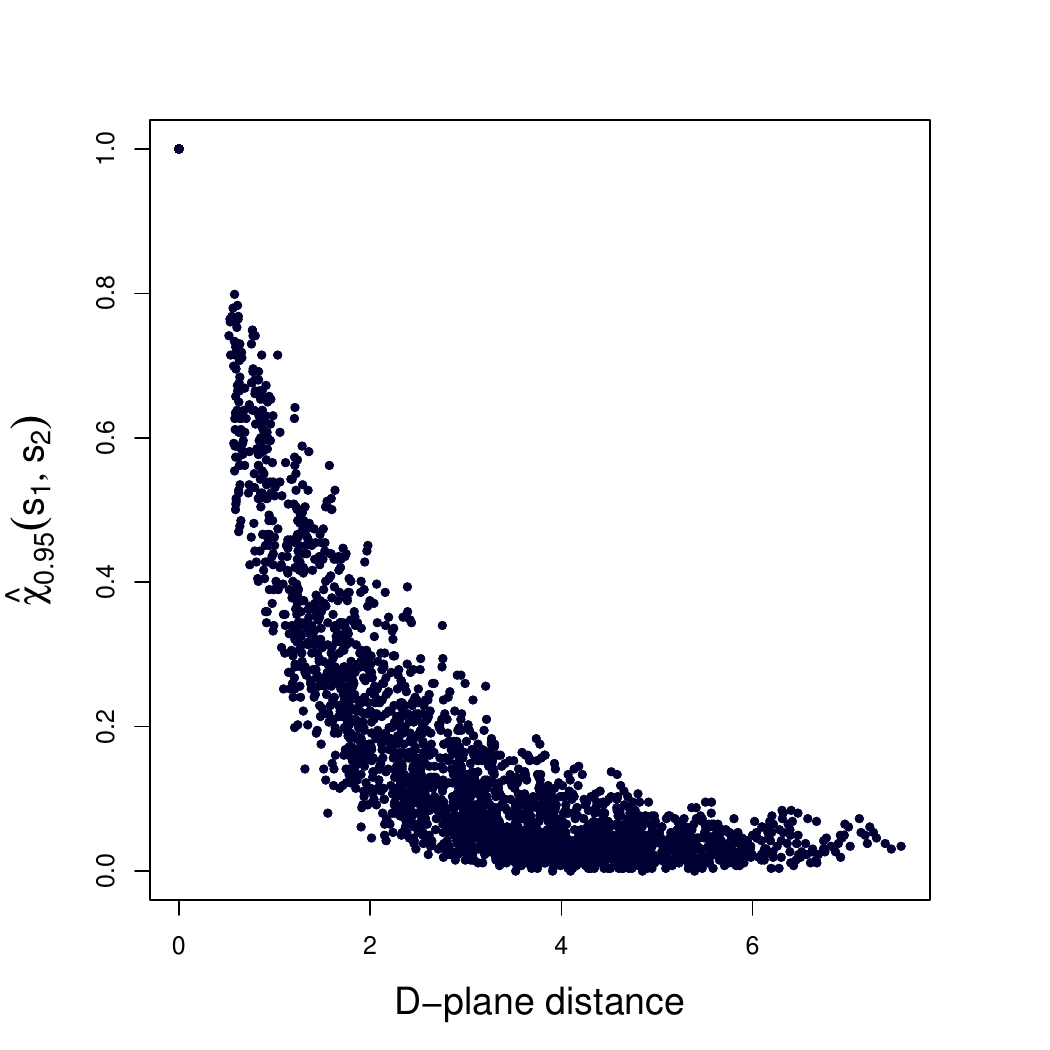}\\
 \caption{Top: G-plane and D-plane, with anchor sites for which $\omega_{r,i} \neq 0$ in~\eqref{eq:tau} highlighted. Bottom: estimates of $\chi_q(s_i,s_j)$ (left) and $\chi_q(\tau(s_i),\tau(s_j))$ (right) for $q=0.95$ plotted against distance in the coordinate system.}
\label{fig:deformation}
\end{figure}

\subsection{Fitted model}
\label{sec:Ausmodelfit}

To select a model, we applied the principles of Section~\ref{sec:modelselection}. Based on likelihood values, Model~3 from Section~\ref{sec:model} with $Z^0$ defined by a Gaussian process $Z_G|Z_G(s_0)=0$, for which $Z_G$ is stationary with powered exponential correlation, marginally transformed to delta-Laplace scale, was deemed to fit the data best. The threshold for observations being taken as extreme at a given site was set at the 97.5\% quantile of the Laplace distribution, which yields an average of 130 exceedances per site. The fit of the model indicates asymptotic independence of the data and independence at long range, which is supported by estimates from a pairwise fit of the model, displayed in Figure~\ref{fig:pairs} of Appendix~\ref{app:supp}. Figure~\ref{fig:pairs} supports parameterizing $\delta$ to decrease from 2 to 1 as distance increases, and so we set $\delta(s-s_0) = 1+\exp\{-(\|s-s_0\|/\delta_1)^{\delta_2}\}$. Parameter estimates are displayed in Table~\ref{tab:param}; Table~\ref{tab:paramdesc} in Appendix~\ref{app:supp} recalls the definition of each parameter for convenience. Estimates of $\beta$ were constrained to be less than one, but are very close. The uncertainty in the parameter estimates is visualized in Figure~\ref{fig:bootstrap} of Appendix~\ref{app:supp}. To produce this figure, we used a stationary bootstrap procedure \citep{PolitisRomano94} to account for the effect of temporal dependence, resampling entire fields with block lengths following a geometric distribution with mean length ten. The maximum composite likelihood estimates from the full dataset have then been subtracted to create a common scale.

% latex table generated in R 3.4.4 by xtable 1.8-2 package
% Wed Oct  9 11:23:24 2019
\begin{table}[ht]
\centering
\caption{Parameter estimates to two decimal places.}
\begin{tabular}{rrrrrrrrrr}
  \hline
&$\kappa$  & $\lambda$ & $\beta$ & $\phi$ & $\mu$ &$\delta_1$ & $\delta_2$ & $\nu$ & $\sigma$   \\\hline
  \hline
All & 1.82 & 1.33 & 1.00 & 2.01 & $-0.08$ & 1.08 & 1.74 & 1.89 & 0.88 \\ 
$Z$ only &-- & -- & -- &  2.06 & -- & 1.36 & 2.31 & 1.87 & 0.97 \\ \hline
El Ni\~no & 1.76 & 1.30 & 1.00 & 2.07 & $-0.06$ & 0.84 & 1.36 & 1.87 & 0.86 \\ 
  $Z$ only & --&--&-- &2.16 & -- & 1.28 & 2.10 & 1.85 & 0.97\\\hline
La Ni\~na & 1.80 & 1.32 & 1.00 & 2.02 & $-0.03$ & 0.99 & 1.65 & 1.87 & 0.87 \\ 
  $Z$ only & --&--&-- & 2.12 & --& 1.37 & 2.43 & 1.84 & 0.99 \\\hline
La Nada & 1.82 & 1.33 & 1.00 & 1.99 & $-0.10$ & 1.08 & 1.74 & 1.88 & 0.87 \\ 
  $Z$ only & --&--&-- & 2.07 & --& 1.37 & 2.32 & 1.86 & 0.97 \\ 
   \hline
\end{tabular}
\label{tab:param}
\end{table}

The assumed model implies that the collection of residual processes $\{Z^j(s)\}_{j=1}^{72}$ have a distribution defined by
\begin{enumerate}
 \item Dependence: $Z_G|Z_G(s_j)=0$, where $Z_G$ is a stationary Gaussian process with mean $\mu$, and covariance function $\mathrm{Cov}(s,s+h) = \sigma^2 \exp\{-(\|h\|/\phi)^\nu\}$;
 \item Margins: delta-Laplace distribution with shape $\delta(s-s_0)$, and location and scale parameter to match that of the conditional field $Z_G|Z_G(s_j)=0$.
\end{enumerate}
As mentioned in Section~\ref{sec:modelselection}, one test of model fit therefore is to examine how closely the residual processes obtained by
\[
 Z^j(s) = \{X(s) - \hat{\alpha}(s-s_j)X(s_j)\}/[1+\{\hat{\alpha}(s-s_j)X(s_j)\}^{\hat{\beta}}]
\]
follow this structure. Upon examination of the empirical $Z^j$, we find a lack of fit, which translates into inadequate reproduction of extreme events. This emanates from an incorrect mean structure. For the possibilities outlined in Section~\ref{sec:processZ}, the mean of $Z_G|Z_G(s_j)=0$ is either increasing or decreasing from zero, or identically zero, as $\|s-s_j\|$ increases. However, in practice here the mean increases then decreases again towards zero; see Figure~\ref{fig:pairs}. This is indicative of the long-range independence in this dataset, since if $X(s_k)$ is independent of $X(s_j)$, with $\alpha(s_k-s_j) \approx 0$, then $Z^j(s_k) = X(s_k)$, which has delta-Laplace margins with $\delta=1,\sigma=1,\mu=0$. One remedy for this would be to attempt to parameterize the non-monotonic form. A further possibility in the current context of gridded data, is to extract the fitted $Z^j$ and re-fit the model for these residual processes using the empirical means $(1/n_j)\sum_{i=1}^{n_j} Z^j_i(s_k)$, $k\in\{1,\ldots,72\}\setminus j$. A disadvantage of this approach is that it would not allow simulation at a new location without placing further spatial structure on the means. Where this is not an issue, as here, an advantage is that it can help alleviate symptoms of non-stationarity not accounted for by the deformation, and we thus adopt this approach. Table~\ref{tab:param} also displays the parameter estimates for $Z^j$ where the Gaussian process model with delta-Laplace margins has been re-fitted, using the empirical means. The estimates of $\delta_1, \delta_2$ and $\sigma$ have the largest change, with the effect illustrated in Figure~\ref{fig:pairs}. Notably $\widehat{\sigma}\approx 1$, ensuring the correct margins for $Z^0(s)$ for large $\|s-s_0\|$. 

As a check on the modelling assumptions, we investigate the independence of $X(s_0)|X(s_0)>u$ and $Z^0$. Figure~\ref{fig:KT} in Appendix~\ref{app:supp} displays Kendall's $\tau$ coefficients between $X(s_j)|X(s_j)>u$ and the mean and variance of the corresponding $Z^j(s)$ for all conditioning sites $s_j$. Based on these diagnostics, independence of $X(s_0)|X(s_0)>u$ and $Z^0$ might be an adequate, but certainly not perfect, assumption. Nonetheless, parameter estimates do not change much when the model is fitted at a higher threshold. Furthermore, other diagnostics such as plots of data simulated from the model versus the observed data, such as the examples shown in Figure~\ref{fig:CompareFitPairs}, suggest a reasonable model fit. 

Using the parameter estimates from the re-fitted $Z^j$, and the original estimates for parameters not relating to $Z^j$, we employ the importance sampling techniques of Section~\ref{sec:importance} to estimate the expected number of grid locations exceeding a certain quantile, given that at least one location exceeds that quantile. Figure~\ref{fig:ExpectedExceedances} displays results for quantiles ranging from 97.5\% - 99.9999\%, with empirical values where available; the in-sample agreement appears reasonable, although there is some question about whether the expected number decreases rapidly enough. Such a summary could not easily be calculated using the pairwise methods in \citet{Winteretal16}. To do so, one would have to fit $72\times71=5112$ pairwise models, calculate the $Z^j$ by concatenating empirical residuals, and then using these with the 10,224 estimated parameters to implement the rejection scheme described in Section~\ref{sec:rejection}. Furthermore, the use of only empirical residuals restricts the shape of new events: where sites are effectively independent from the conditioning site, i.e., $a_{s-s_0}(x) \approx 0$, simulated new events will look just like past events in these areas. The figure also displays the estimate from the Brown--Resnick Pareto process, fitted to all processes for which the value $\max_{1\leq j \leq 72}X(s_j)$ exceeds its $97.5\%$ quantile. Whilst this provides a reasonable estimate close to the fitting threshold, it is clearly inappropriate for extrapolation further into the tail.

\subsection{Investigation of El Ni\~no effect}
\label{sec:elnino}

In Australia, the value of ENSO is known to impact the marginal distribution of extreme temperatures at specific locations \citep{Winteretal16}. To investigate the possible effect of ENSO on the spatial extent of high temperature events, separate models were fitted to data from El Ni\~no, La Ni\~na and ``La Nada'' seasons. El Ni\~no (respectively La Ni\~na) events are defined here as those for which the sea surface temperature (SST) anomaly from 1980-2010 mean levels is higher than $+1^\circ$ (respectively lower than $-1^\circ$); La Nada events are the remaining ones. The anomalies were taken from \texttt{https://www.esrl.noaa.gov/psd/gcos\_wgsp/Timeseries/Data/nino34.long.anom.data}. The covariate was defined by averaging monthly SST anomalies over the summer, so that each summer season is either El Ni\~no, La Ni\~na or La Nada. Marginal transformations were made separately for the three categories, whilst dependence parameter estimates are given in Table~\ref{tab:param}. There is some modest deviation from the combined parameter estimates particularly for the El Ni\~no years where $\delta$ decays more rapidly, indicating slightly increased variability in these years. Figure~\ref{fig:ExpectedExceedances} displays estimates of the expected number of exceedances from the models fit to the different ENSO regimes: La Ni\~na years are slightly higher and El Ni\~no slightly lower, but differences do not appear very large. However, differences in marginal quantiles make the practical interpretation less straightforward. In the South-East of Australia particularly, the high quantiles represent hotter temperatures under El Ni\~no than La Ni\~na, and a temperature that represents an exceedance of a marginal 99.5\% quantile in La Ni\~na conditions --- which from Figure~\ref{fig:ExpectedExceedances} occur in around 3.5 grid squares on average --- might be closer to a 97.5\% marginal quantile, or even lower, in El Ni\~no conditions, which Figure~\ref{fig:ExpectedExceedances} suggests might affect 5 or more grid squares on average. The ENSO regime may also affect the temporal persistence of extreme heat events, but this is not captured by our purely spatial models.

\begin{figure}
\centering
 \includegraphics[width=0.45\textwidth]{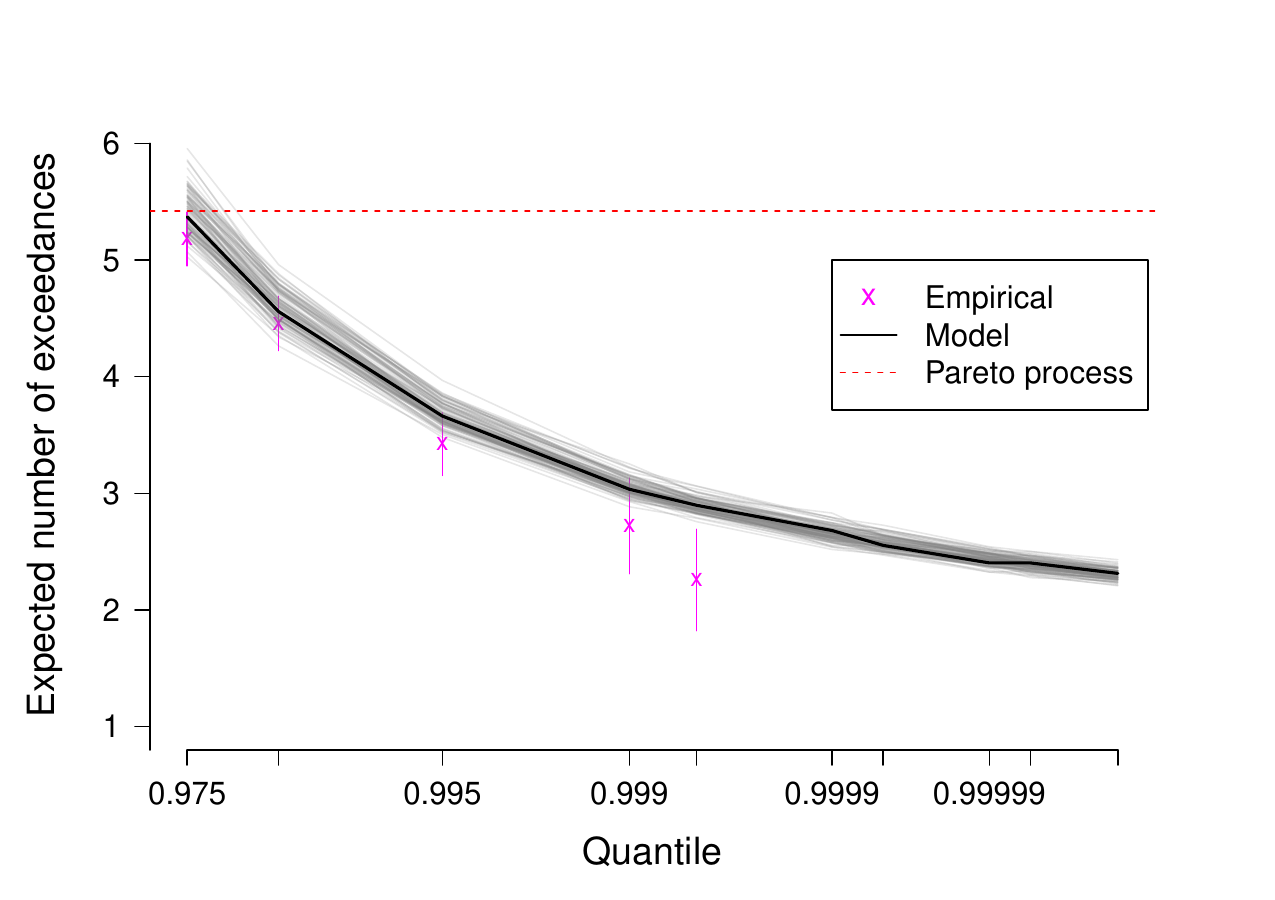}
 \includegraphics[width=0.45\textwidth]{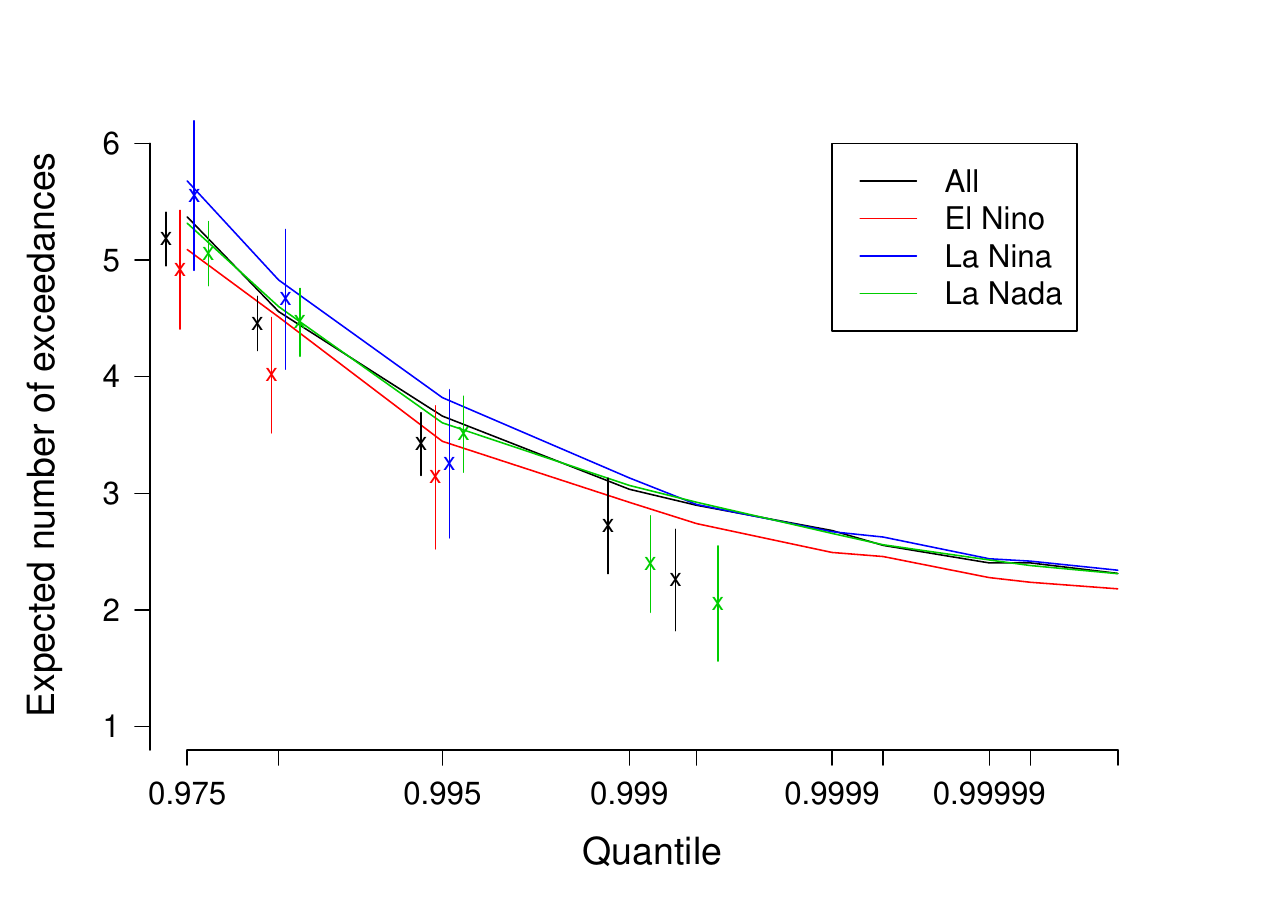}
 \caption{Expected number of exceedances of a particular marginal quantile given that there is at least one exceedance of that quantile somewhere over Australia. Points marked $\times$ are this quantity estimated directly from the data, with bars giving approximate 95\% confidence intervals; there are no data to estimate directly above the 99.95\% quantile. Left: thick solid line represents model-based estimation, with estimates from bootstrap samples in grey. Right: different colours denote estimates from the models fitted to all data, as well as that stratified by ENSO covariate.}
\label{fig:ExpectedExceedances}
\end{figure}

\section{Discussion}
The conditional approach to spatial extreme value analysis offers a number of advantages: flexible, asymptotically-motivated dependence structures that can capture asymptotic (in)dependence; the models can be fit in reasonably high dimensions; and conditional simulation at unobserved locations is simple. The principal drawback of the approach lies in the more complicated interpretation of what constitutes ``the model'' for a given dataset. When conditioning only on a single site being extreme, interpretation is straightforward. When wishing to condition on any of a set of sites being extreme, we need to combine these individual models in a way that entails no clearly defined overall model for the process given that it is extreme somewhere over space. However, for the purpose of \emph{inference} on many quantities of interest, the importance sampling techniques ensure that this is not an issue. 

Owing to computational limitations, there have been relatively few attempts at high-dimensional inference for extremes. In general, higher-dimensional problems do not only present computational issues, but they are often accompanied by additional modelling challenges. The spatial domain of interest is likely to be larger, leading to potentially more diverse behaviour in the dependence. In the Australian temperature example, we observed non-stationarity, and independence at longer range. These issues are less likely to arise when focusing on smaller areas, and these complexities should be kept in mind as spatial extremes moves into a higher-dimensional phase.

A natural next step is to extend this approach in to a multivariate or space-time setting. For example, convergence~\eqref{eq:condlimit} can be generalized by replacing all locations $s\in\mathcal{S}$ by $(s,t) \in \mathcal{S}\times\mathcal{T} \subset \mathbb{R}^2 \times \mathbb{R}$, with the conditioning location $(s_0,t_0)$. For modelling purposes, there are additional considerations of the dependence regime in both space and time, and how these interact. However, in principle one can have asymptotic (in)dependence in both space and time, with potentially different behaviour in the different dimensions. This spatio-temporal setting is the focus of recent work in \citet{SimpsonWadsworth21}.

Further extensions to this approach allow for different forms of the functions $a_{s-s_0}, b_{s-s_0}$ or residual process $Z^0$. Building on an early draft of this paper \citet{Shooteretal19,Shooteretal21} have proposed alternative parametric models for the normalization functions, while \citet{Simpsonetal20} use a spline for the $a_{s-s_0}$ term, avoiding an imposed parametric form. The latter also explore the use of Gaussian Markov random field structures for $Z^0$ in order to allow for a substantially larger number of observation locations, through the machinery of the integrated nested Laplace approximation (INLA).

\subsection*{Acknowledgements}
We thank Simon Brown at the UK Met Office for the data analyzed in Section~\ref{sec:Australia}. J. Wadsworth gratefully acknowledges funding from EPSRC grant EP/P002838/1.

\subsection*{Data and code}
The data and code for the analysis of Section~\ref{sec:Australia} are available as Supplementary Material. The most recent data can be obtained from~\texttt{https://www.metoffice.gov.uk/hadobs/hadghcnd/}, subject to the conditions detailed at the URL.

\appendix
\section{Additional results and proofs}
\label{app:proof}

Before the proof of Proposition~\ref{prop:sameab}, Lemma~\ref{lem:convflexibility} specifies conditions under which there is some flexibility in the normalization leading to convergence.

\begin{lemma}
\label{lem:convflexibility}
Suppose that for $\ba(v)$ and $\bb(v)$ with twice-differentiable components $a_l, b_l$ satisfying $a_l'(v) \sim \alpha_l$, $a''_l(v) = o(1)$, $b_l'(v)/b_l(v) = o(1)$, as $v\to\infty$, for $l=1,\ldots, d$,
 \begin{align*}
  \P\left(\frac{\bm{V}-\ba(V_0)}{\bb(V_0)} \leq \bm{z}~\Big|~V_0 = v\right) = \P\left(\frac{\bm{V}-\ba(v)}{\bb(v)} \leq \bm{z}~\Big|~V_0 = v\right) \to G(\bm{z}),\qquad v\to\infty, 
 \end{align*}
 and that all first and second order partial derivatives converge, i.e.
  \begin{align*}
  \frac{\partial}{\partial z_k}\P\left(\frac{\bm{V}-\ba(v)}{\bb(v)} \leq \bm{z}~\Big|~V_0 = v\right) &\to  \frac{\partial}{\partial z_k}G(\bm{z}),\\
  \frac{\partial^2}{\partial z_k \partial z_l}\P\left(\frac{\bm{V}-\ba(v)}{\bb(v)} \leq \bm{z}~\Big|~V_0 = v\right) &\to  \frac{\partial^2}{\partial z_k \partial z_l}G(\bm{z}),\qquad v\to\infty. 
 \end{align*}
 Then
 \begin{enumerate}[(i)]
  \item for $\bm{h}^1(v) = (h_1^1(v),\ldots,h^1_d(v))^\top$, $\bm{h}^2(v) = (h_1^2(v),\ldots,h_d^2(v))^\top$, with $h_l^1(v)=o(1)$, $h_l^2(v)=o(1)$, for all $l=1,\ldots,d$,
  \begin{align*}
  \P\left(\frac{\bm{V}-\ba(v)}{\bb(v)[1+\bm{h}^1(v)]} + \bm{h}^2(v) \leq \bm{z}~\Big|~V_0 = v\right) \to G(\bm{z}),\qquad v\to\infty.
 \end{align*}
 \item for $\bm{h}^1$ as above but $h^2_l(v) = O(1)$, i.e., some components may be asymptotically non-zero constants, whilst some may converge to zero, 
  \begin{align*}
  \P\left(\frac{\bm{V}-\ba(v)}{\bb(v)[1+\bm{h}^1(v)]} + \bm{h}^2(v) \leq \bm{z}~\Big|~V_0 = v\right) \to G(\bm{z} - \lim_{v\to\infty}\bm{h}^2(v)),\qquad v\to\infty.
 \end{align*}
 \end{enumerate}
 
\end{lemma}

\begin{proof}
\begin{enumerate}[(i)]
 \item 
 Write
  \begin{align*}
  \P\left(\frac{\bm{V}-\ba(v)}{\bb(v)} \leq \bm{z}~\Big|~V_0 = v\right) &= F_{V|V_0}(\bb(v) \bm{z} + \ba(v) | v),\\ \P\left(\frac{\bm{V}-\ba(v)}{\bb(v)[1+\bm{h}^1(v)]} + \bm{h}^2(v) \leq \bm{z}~\Big|~V_0 = v\right) &= F_{V|V_0}(\bb(v)\bm{z} + \ba(v) + \bb(v)\bm{h}^1(v)\bm{z} - \bb(v)\bm{h}^2(v)  +O(\bb(v)\bm{h}^1(v)\bm{h}^2(v))| v).
 \end{align*}
 Now write $\bm{g}(v) = \bb(v)\bm{h}^1(v)\bm{z} - \bb(v)\bm{h}^2(v)  +O(\bb(v)\bm{h}^1(v)\bm{h}^2(v))$, with $g_l(v) =o(b_l(v))$, and consider the Taylor expansion
   \begin{align}
   \label{eq:mvtaylor}
  F_{V|V_0}(\bb(v) \bm{z} + \ba(v) + \bm{g}(v)| v) = F_{V|V_0}(\bb(v)\bm{z} + \ba(v)| v) + \nabla F_{V|V_0}(\bb(v) \bm{z} + \ba(v)| v)^\top \bm{g}(v) + O(\max(\vee_l h_1^l(v),\vee_l h_2^l(v))^2).
 \end{align}
 The components of $\nabla F_{V|V_0}(\bb(v)\bm{z} + \ba(v)| v)^\top$ are $F^{(l)}_{V|V_0}(\bb(v) \bm{z} + \ba(v)| v)$, where $F^{(l)}_{V|V_0}(\bm{x}|v) = \frac{\partial}{\partial y_l}F_{V|V_0}(\bm{y}|v)|_{\bm{y}=\bm{x}}$. The convergence 
 \begin{align*}
  \frac{\partial}{\partial z_l}\P\left(\frac{\bm{V}-\ba(V_0)}{\bb(V_0)} \leq \bm{z}~\Big|~V_0 = v\right) \to  \frac{\partial}{\partial z_l}G(\bm{z})
 \end{align*}
 is equivalent to 
 \begin{align*}
  F^{(l)}_{V|V_0}(\bb(v) \bm{z} + \ba(v)| v)b_l(v) \to \frac{\partial}{\partial z_l}G(\bm{z}),
 \end{align*}
 hence
 \begin{align}
 \label{eq:partialconv}
  F^{(l)}_{V|V_0}(\bb(v)\bm{z} + \ba(v)| v)=\frac{\partial}{\partial z_l}G(\bm{z}) \{b_l(v)\}^{-1}[1+o(1)],
 \end{align}
 with a similar approach for the mixed partial derivatives. Substituting~\eqref{eq:partialconv} into~\eqref{eq:mvtaylor} yields that the second term in the expansion is $O(\max(\vee_lh^1_l(v),\vee_lh^2_l(v)))$, whilst the next order term is $O(\max(\vee_lh^1_l(v),\vee_lh^2_l(v))^2)$.

 \item  The proof is very similar to part~(i), except the Taylor expansion is about $\bb(v)(\bm{z}-\lim_{v\to\infty}\bm{h}^2(v)) + \ba(v)$.
\end{enumerate}

\end{proof}

\begin{proof}[Proof of Proposition~\ref{prop:sameab}]

Let $q^\star_x=([\min\{\min_l(a_l(x)+b_l(x) z_l),x\}]/\lambda)_+$ with $y_+=\max(y,0)$. We have
\begin{align*}
 \lim_{x\to\infty} \P\left(\frac{\bm{X}-\ba(X_0)}{\bb(X_0)} \leq \bm{z}~\Big|~X_0 = x\right) &=  \lim_{x\to\infty}\int_0^{q^\star_x} \P\left(\frac{\bm{X}-\ba(X_0)}{\bb(X_0)} \leq \bm{z}, Q=q ~\Big|~X_0 = x\right) \,\mathrm{d}q\\
 &=  \lim_{x\to\infty}\int_0^{q^\star_x} \P\left(\frac{\bm{X}-\ba(x)}{\bb(x)} \leq \bm{z} ~\Big|~ Q=q, X_0 = x\right) f_{Q|X_0}(q|x) \,\mathrm{d}q\\
 &=  \lim_{x\to\infty}\int_0^{q^\star_x} \P\left(\frac{\bm{V} + \lambda q  -\ba (v(x)+\lambda q )}{\bb(v(x)+\lambda q)} \leq \bm{z} ~\Big|~ V_0=v(x),Q=q\right)\\ & \qquad\qquad\qquad\times f_{Q|X_0}(q|x) \,\mathrm{d}q\\
  &=  \lim_{x\to\infty}\int_0^{q^\star_x} \P\left(\frac{\bm{V} -\ba(v(x)) +(1-\bm{\alpha})\lambda q + o(1)}{\bb(v(x))[1+O(\nabla \bb(v(x))/\bb(v(x)))]} \leq \bm{z} ~\Big|~ V_0=v(x),Q=q\right)\\
  &\qquad\qquad\qquad \times f_{Q|X_0}(q|x) \,\mathrm{d}q,
\end{align*} with $v(x) = x-\lambda q $ and $f_{Q|X_0}(q|x)$ the conditional density of $Q|X_0=x$. The integrand is dominated by $f_{Q|X_0}(q|x)$, and
\begin{align*}
 f_{Q|X_0}(q|x)  = \frac{f_{Q,V_0}(q,x-\lambda q )}{f_{X_0}(x)}  &= \frac{e^{-q}e^{-(x-\lambda q )}\I(q>0)\I(x-\lambda q >0)}{e^{-x}[(1-\lambda)^{-1} - (1-\lambda)^{-1}e^{-(1/\lambda-1)x}]} \to (1-\lambda)e^{-(1-\lambda)q}\I(q>0),\qquad x\to\infty;
\end{align*}
note that $f_{Q|X_0}(q|x)$ is uniformly bounded by the integrable function $g(q)=(1-\lambda)e^{-(1-\lambda)q}[1-e^{-(1/\lambda-1)R}]^{-1}$ for all $x>R$.
By Lemma~\ref{lem:convflexibility}, as $v \to \infty$,
\begin{align*}
  \P\left(\frac{\bm{V} -\ba(v)+(1-\bm{\alpha})\lambda q + o(1)}{\bb(v) [1+O(\nabla \bb(v)/\bb(v))]} \leq \bm{z} ~\Big|~ V=v,Q=q\right) \to  G\left(\bm{z}+\frac{(\bm{\alpha}-1)\lambda q}{\lim_{v\to\infty} \bb(v)}\right),
\end{align*}
which simplifies when $\lim_{v\to\infty}b_l(v)=\infty$ for all $l$. Dominated convergence then yields the result stated.
\end{proof}

\begin{lemma}[Convergence of Gaussian partial derivatives]
\label{lem:Gausspd}
 Suppose $(\bm{Y},Y_0)^\top$ follows a $(d+1)$-dimensional Gaussian distribution, and let $(\bm{V},V_0)^\top = T((\bm{Y},Y_0)^\top) \in \mathbb{R}^{d+1}_+$ be a random vector with unit exponential margins and Gaussian copula. Denote by $\bm{\rho}_0 > \bm{0}$ the $d$-vector of correlation parameters between $\bm{Y}$ and $Y_0$. Then for $\ba(v) = \bm{\rho}_0^2 v$, $\bb(v) = 1+(\bm{\rho}_0v)^{1/2}$ and any $r \leq d$,
   \begin{align*}
  \frac{\partial^{r}}{\partial z_1\cdots \partial z_r}\P\left(\frac{\bm{V}-\ba(v)}{\bb(v)} \leq \bm{z}~\Big|~V_0 = v\right) &\to \frac{\partial^{r}}{\partial z_1\cdots \partial z_r}G(\bm{z}).\\
 \end{align*}
\end{lemma}
\begin{proof}
 We can express
 \begin{align*}
  \P\left(\frac{\bm{V}-\ba(v)}{\bb(v)} \leq \bm{z}~\Big|~V_0 = v\right) &=  \P\left[\bm{Y} \leq T^{-1}\{\bb(T(y))\bm{z} + \ba(T(y))\}~\Big|~T(Y_0) = T(y)\right]\\
  & = \Phi_d \left[T^{-1}\{\bb(T(y))\bm{z} + \ba(T(y))\}- \bm{\rho}_0 y; \Sigma_{0}\right]
 \end{align*}
where $\Phi_d(\cdot;\Sigma_0)$ is the cdf of the centred $d$-variate Gaussian with covariance matrix $\Sigma_0 = (\rho_{k,l}-\rho_{k,0}\rho_{l,0})_{1\leq k,l\leq d}$
Taking the derivative yields
\begin{align}
   &\frac{\partial^{r}}{\partial z_1\cdots \partial z_r} \P\left[\bm{Y} \leq T^{-1}\left\{\bb(T(y))\bm{z} + \ba(T(y))\right\}~\Big|~T(Y_0) = T(y)\right]\notag \\ &\qquad\qquad = \Phi_d^{(1:r)}  \left[T^{-1}\{\bb(T(y))\bm{z} + \ba(T(y))\}; \Sigma_{0}\right] \prod_{k=1}^r \frac{\partial}{\partial z_k}T^{-1}\{b_{k}(T(y))z_k + a_{k}(T(y))\}, \label{eq:deriv}
\end{align}
where $\Phi_d^{(1:r)}(\cdot| \Sigma_0)$ is the $r$th-order mixed partial derivative of $\Phi_d$. The components 
\begin{align}
\frac{\partial}{\partial z_k}T^{-1}\{b_{k}(T(y))z_k + a_{k}(T(y))\} =  \frac{b_{k}(T(y))}{T'\left[T^{-1}\{b_{k}(T(y))z_k + a_{k}(T(y))\}\right]}.\label{eq:derivcomp}
\end{align}
Now $T(y) = -\log\{1-\Phi(y)\} = y^2/2 +O(\log y)$, $y \to \infty$, whilst $T^{-1}(x) =(2x)^{1/2} + O(\log x/x^{1/2})$, $x \to \infty$, and $T'(y) = \phi(y)/\{1-\Phi(y)\} \sim y$, $y \to \infty$. Further, $b_{k}(T(y)) = 1+\rho_{0,k}y/\sqrt{2} + O(\log y)$, whilst $a_{k}(T(y)) =\rho_{0,k}^2 y^2/2 + O(\log y)$. Combining these,
\begin{align*}
 T^{-1}\{b_{k}(T(y))z_k + a_{k}(T(y))\} = \rho_{0,k} y +z_k/\sqrt{2} +o(1),
\end{align*}
and equation~\eqref{eq:derivcomp} converges to $1/\sqrt{2}$, whilst~\eqref{eq:deriv} converges to
\begin{align*}
\Phi_d^{(1:r)}\left(\bm{z}/\sqrt{2}; \Sigma|_{0}\right) 2^{-r/2},
\end{align*}
which is the $r$th-order mixed partial derivative of the Gaussian limit distribution. 
\end{proof}

\begin{remark}
\label{rmk:Gauss2deriv}
 For application of Proposition~\ref{prop:sameab}, we also require convergence of the second derivative $\partial^2/(\partial z_k)^2$. Iteration of the above manipulations leads to a conclusion that the derivative converges to $\Phi_d^{(kk)}(\bm{z}/\sqrt{2}; \Sigma_{0})/2$, with $\Phi_d^{(kk)}(\bm{z}) = \partial^2 \Phi_d(\bm{y})/(\partial y_k)^2|_{\bm{y}=\bm{z}}$.
\end{remark}

\begin{prop}
\label{prop:abcondns}
 Suppose that $(X_1,X_2)$ have identical margins with infinite upper endpoint and
 \begin{enumerate}
  \item $\lim_{x \to \infty} \P(X_1>x|X_2>x) =0$
  \item $\lim_{t \to \infty} \P\left[\{X_1 - a(X_2)\}/b(X_2) >z |X_2>t\right] = G(z)$,
 \end{enumerate}
with $G$ a non-degenerate distribution function for which $\lim_{z\to \infty}G(z) = 1$. Then if there exists $t_0$ such that $a(t)>0$ is monotonically non-decreasing in $t$ for all $t>t_0$
\begin{enumerate}[(i)]
 \item If $G(z) \in (0,1)$ for some $z \in (0,\infty)$, $a(t)<t$ for all sufficiently large $t$.
 \item If $G(z) \in (0,1)$ for all $z \in (0,\infty)$, then $b(t) = o(t)$, $t\to\infty$.
\end{enumerate}
\end{prop}

\begin{proof}
(i). When $z>0$, 
  \begin{align*}
   0 \leq \P\left(X_1 >a(X_2) + b(X_2)z |X_2>t\right) &\leq \P\left(X_1 >a(X_2) |X_2>t\right)\leq \P\left(X_1 >a(t) |X_2>t\right),
  \end{align*}
for sufficiently large $t$. But since $\P\left(X_1 >t |X_2>t\right) \to 0$, we must have $a(t)<t$ for all large $t$ to get convergence to $G(z) \in (0,1)$. \\
(ii). If $b(t)$ is constant or decreasing then clearly the statement holds, so suppose it is increasing as $t \to \infty$.  Similarly to above
  \begin{align*}
   0 \leq \P\left(X_1 >a(X_2) + b(X_2)z |X_2>t\right) &\leq \P\left(X_1 >b(X_2)z |X_2>t\right) \leq \P\left(X_1 >b(t)z |X_2>t\right)
  \end{align*}
 for $z>0$ and large $t$. Now if $b(t) \sim ct$ for some $c>0$ then for any $z>1/c$, and sufficiently large $t$, $0< \P\left(X_1 >b(t)z |X_2>t\right) \leq \P\left(X_1 >t |X_2>t\right) \to 0.$ But since $G(z) \in (0,1)$ for any $z \in (0,\infty)$, we must have $b(t) = o(t)$. 
\end{proof}

%\newpage
\section{Supporting information for Section~\ref{sec:Australia}}
\label{app:supp}

\subsection{Model information}
\begin{table}[h]\centering
\caption{Description of parameters in fitted model}
\label{tab:paramdesc}
 \begin{tabular}{ll}\hline
  Parameter & Description \\\hline
  $\kappa$  &Shape parameter in $\alpha(s-s_0)$ (eq.~\eqref{eq:alpha})\\
  $\lambda$& Scale parameter in $\alpha(s-s_0)$ (eq.~\eqref{eq:alpha})\\
  $\beta$ &Power in $b_{s-s_0} = 1+a_{s-s_0}(x)^\beta$\\
  $\phi$ &Scale parameter in Gaussian correlation\\
  $\nu$ &Shape parameter in Gaussian correlation\\
  $\sigma$ &Scale parameter of $Z_G$ used to determine structure of $Z^j$\\
  $\mu$ &Mean parameter of $Z_G$ used to determine structure of $Z^j$\\
  $\delta_1$ &Scale parameter in $\delta(s-s_0) = 1+ \exp\{-(\|s-s_0\|/\delta_1)^{\delta_2}\}$\\
  $\delta_2$ &Shape parameter in $\delta(s-s_0) = 1+ \exp\{-(\|s-s_0\|/\delta_1)^{\delta_2}\}$\\
  \hline
 \end{tabular}

\end{table}

\subsection{Diagnostic plots and parameter uncertainty}

\begin{figure}[H]
 \centering
 \includegraphics[width=0.4\textwidth]{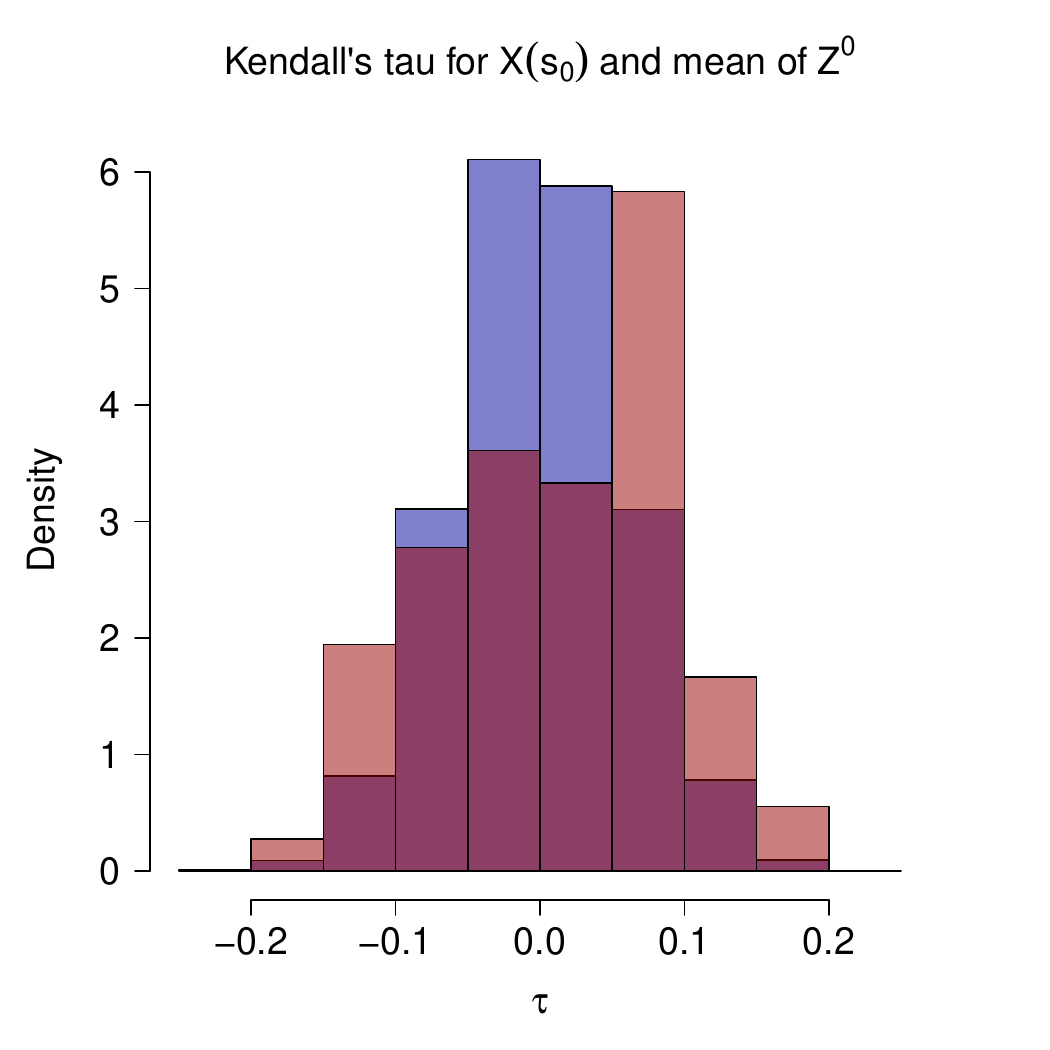} % file name with a 1 at the end is latest version; with a 2 at the end is the existing version!
 \includegraphics[width=0.4\textwidth]{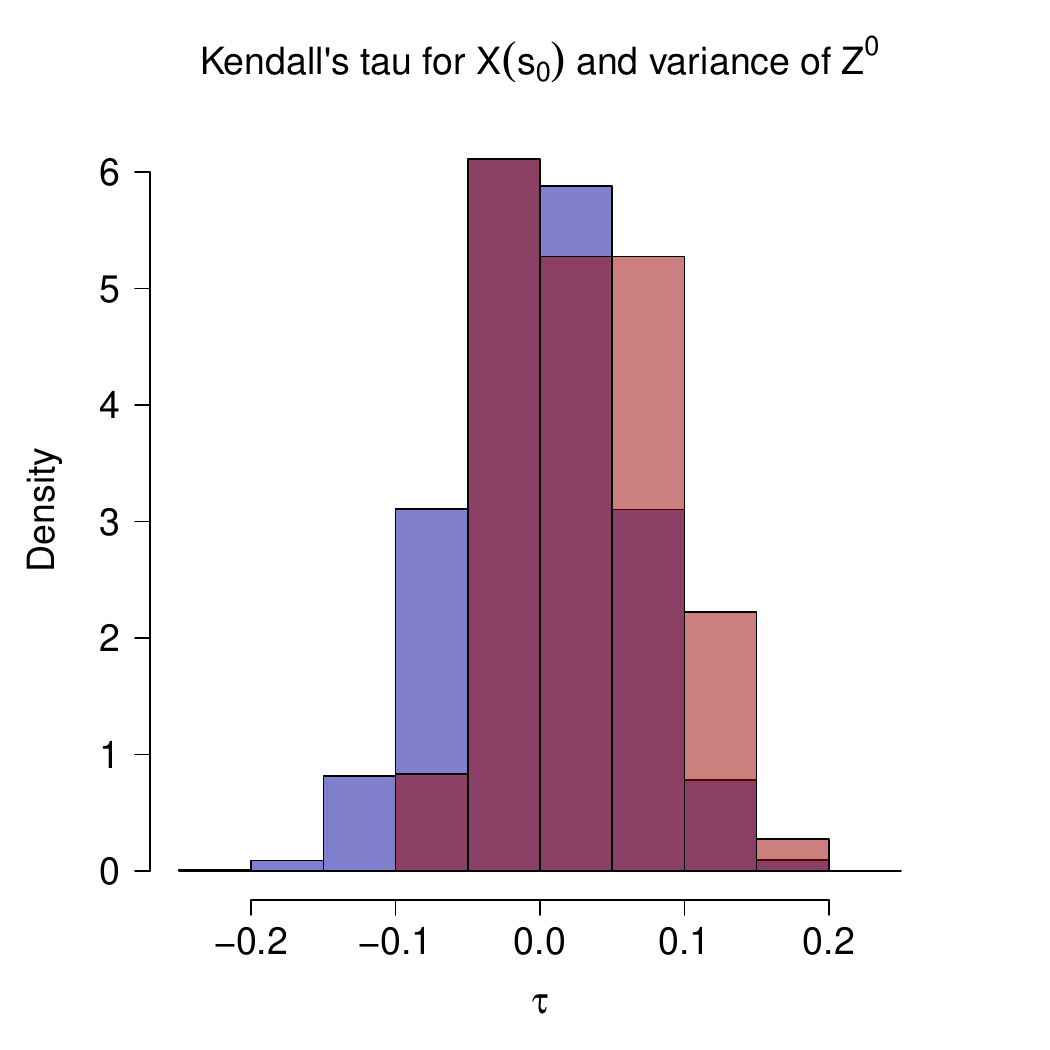} % file name with a 1 at the end is latest version; with a 2 at the end is the existing version!
 \caption{Histograms (red) of Kendall's $\tau$ coefficients for: $X^i(s_j)|X^i(s_j)>u$ and the mean (left) and variance (right) of $Z^j_i(s)$, $i=1,\ldots, n_j$ for the 72 sites, with average sample size $n_j \approx 130$.  For comparison, in blue is the histogram of the approximate null Kendall's $\tau$ distribution obtained from 10000 samples of 130 independent bivariate datapoints to give an impression of the null distribution.}
 \label{fig:KT}
\end{figure}

\begin{figure}[H]
 \centering
 \includegraphics[width=0.6\textwidth]{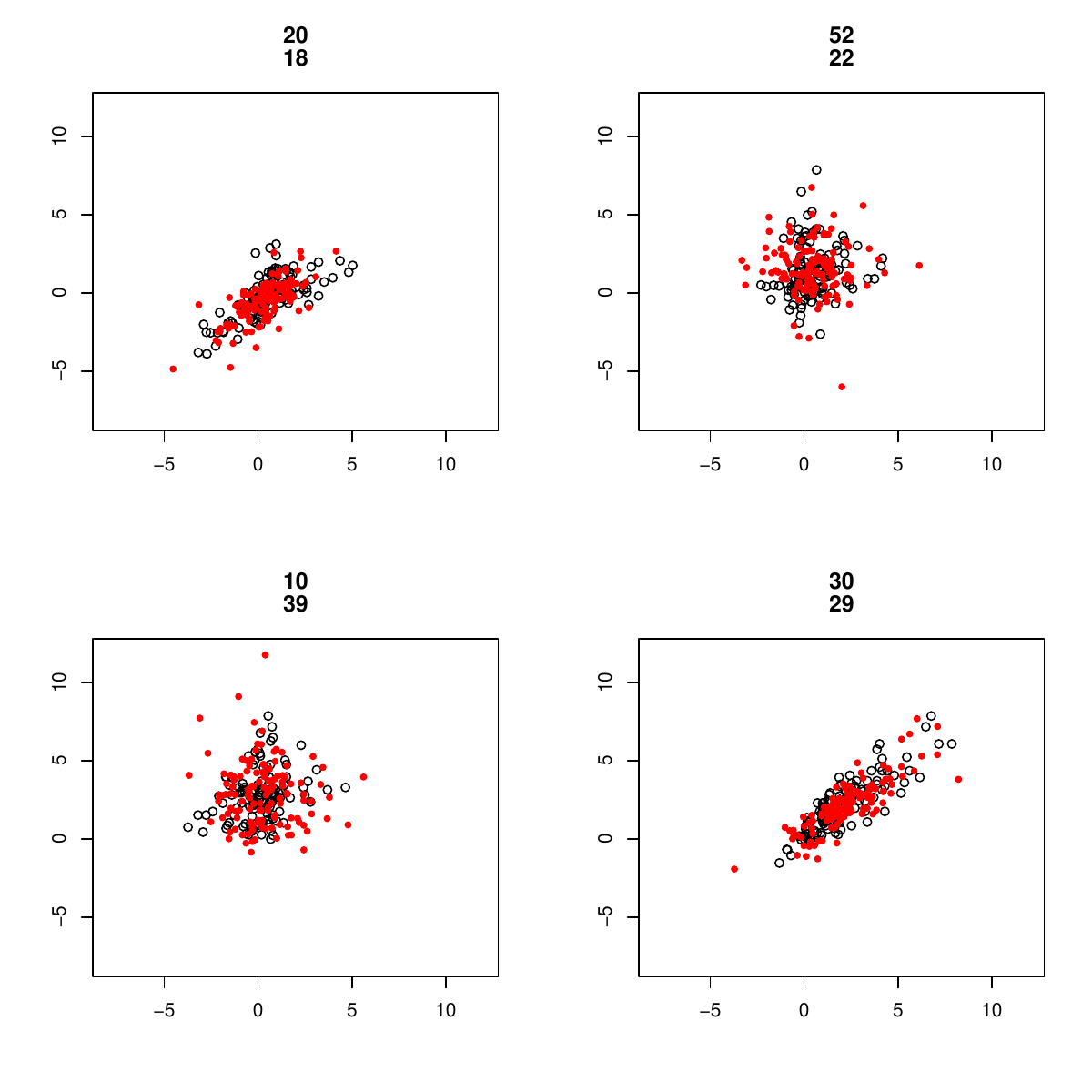}
 \caption{Example pairwise plots of observed data (black unfilled circles) and corresponding data simulated from the model (red filled circles).  The site pairings are given in the figure heading, and the conditioning site is 48, which is not one of the sites appearing in the pairs. The observed data correspond to observations at the two sites when the observation at the conditioning site exceeds its threshold. The simulated data are generated from the model conditioning on the same site exceeding the same threshold.}
 \label{fig:CompareFitPairs}
\end{figure}

\begin{figure}[H]
 \centering
 \includegraphics[width=0.25\textwidth]{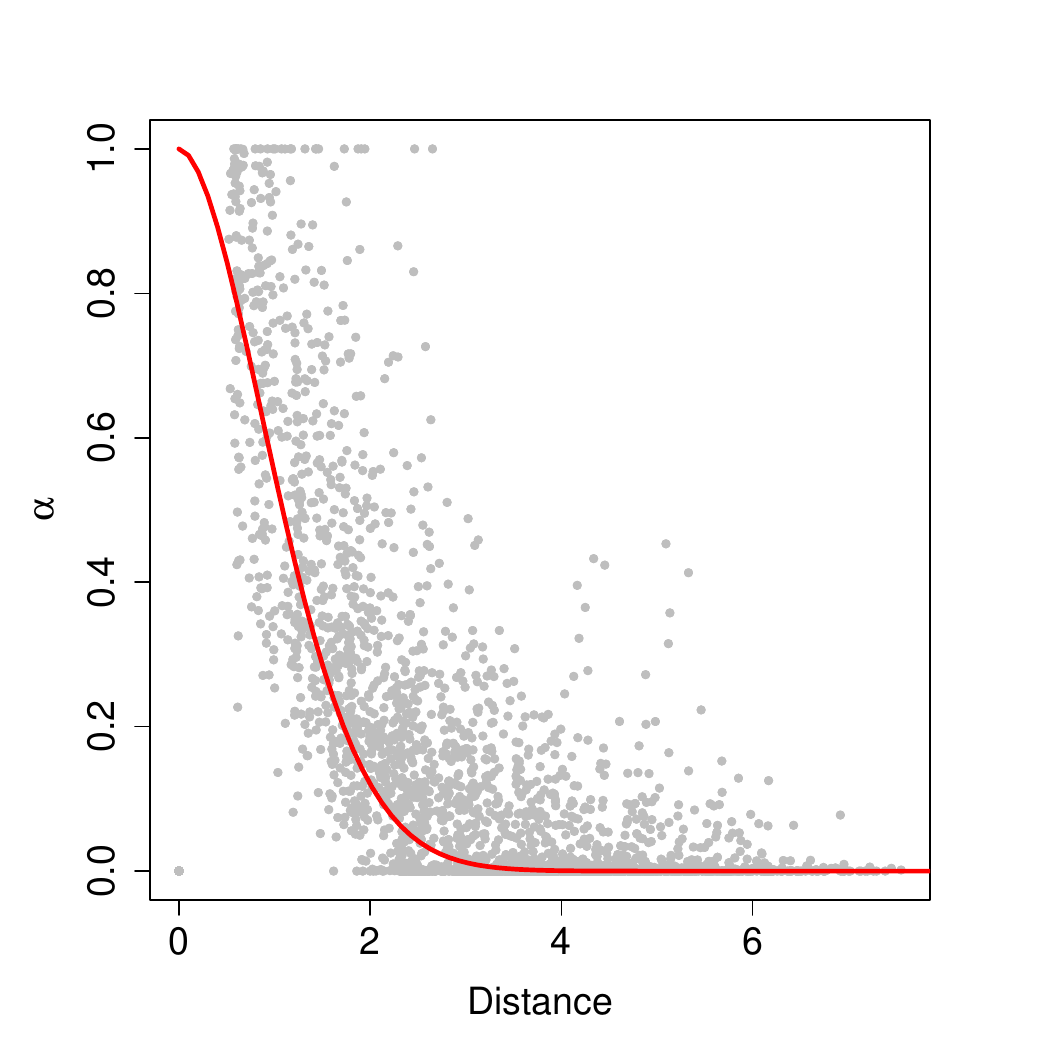}
 \includegraphics[width=0.25\textwidth]{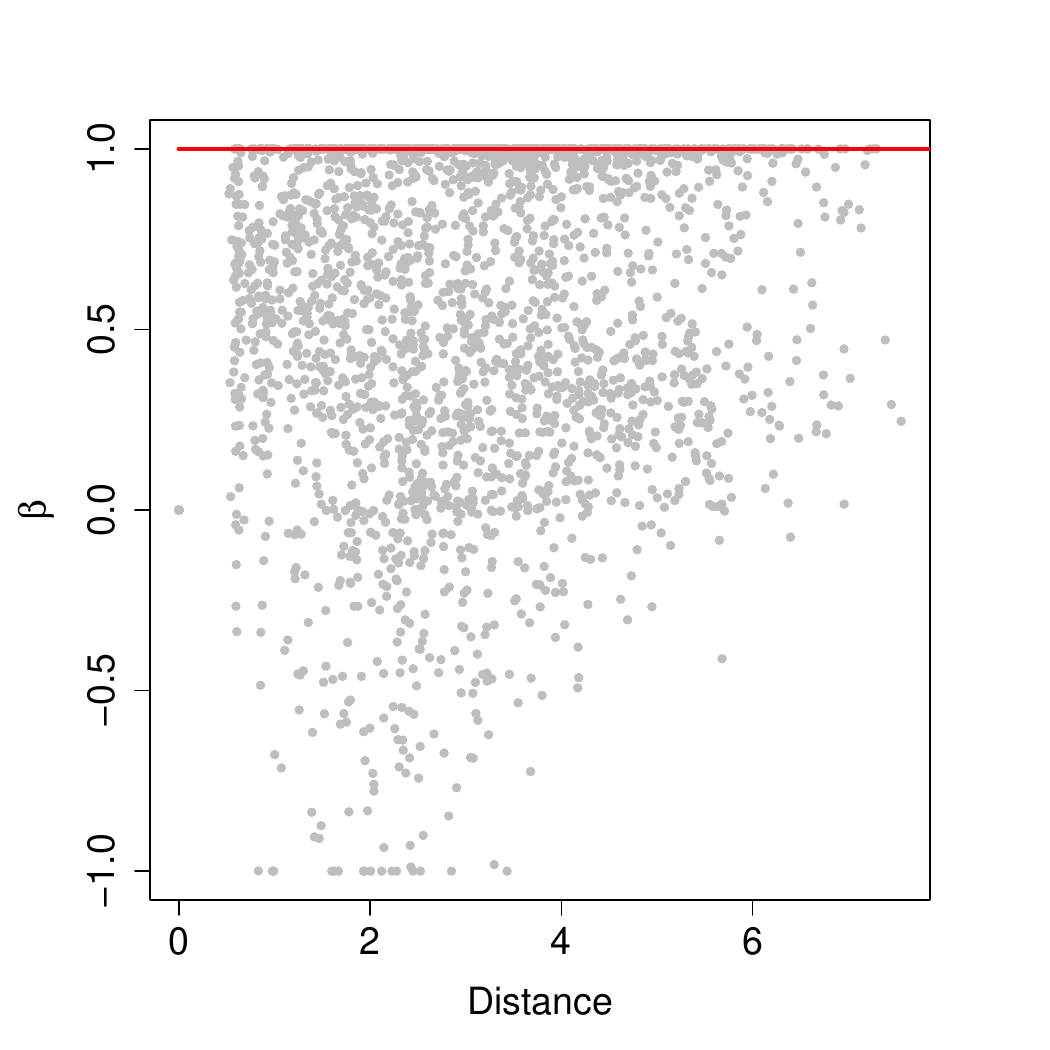}
 \includegraphics[width=0.25\textwidth]{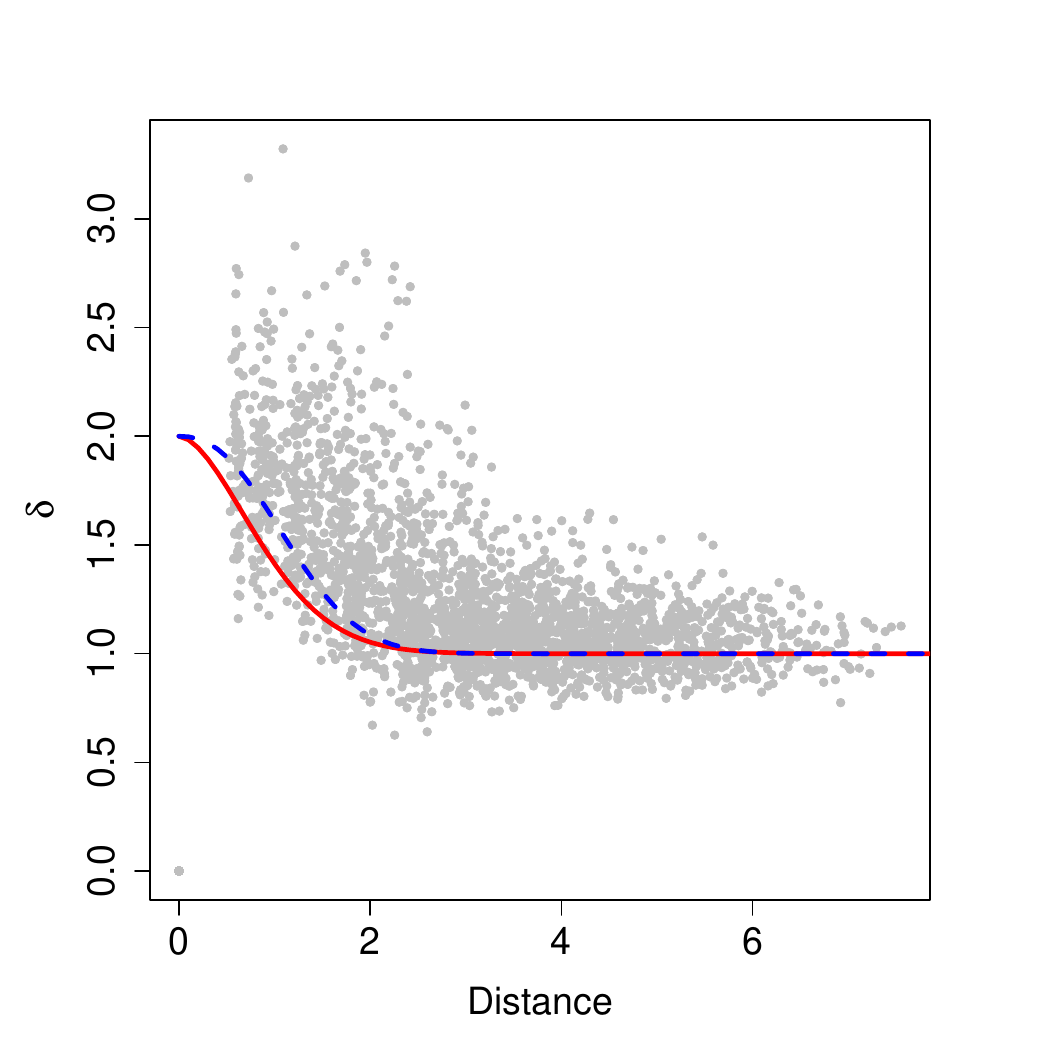}\\
 \includegraphics[width=0.25\textwidth]{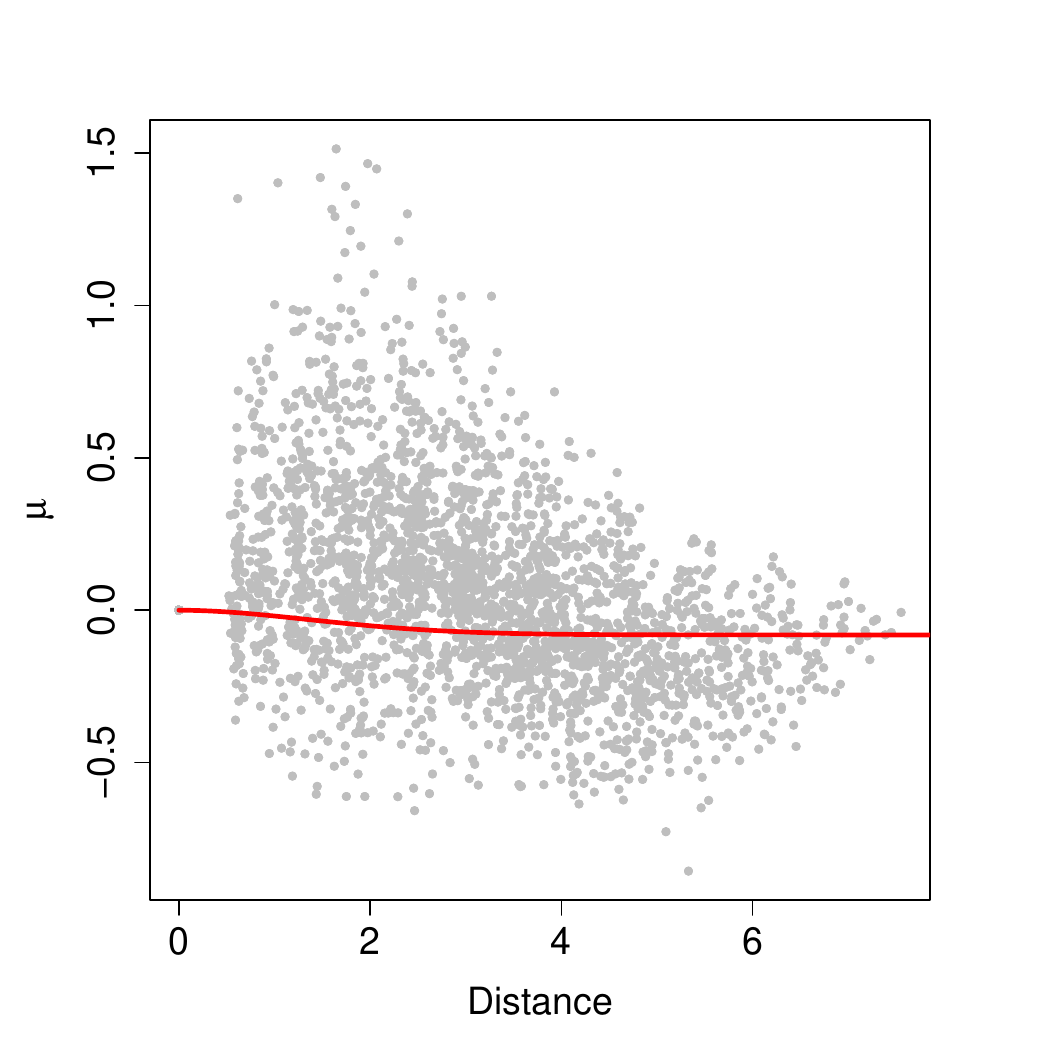}
 \includegraphics[width=0.25\textwidth]{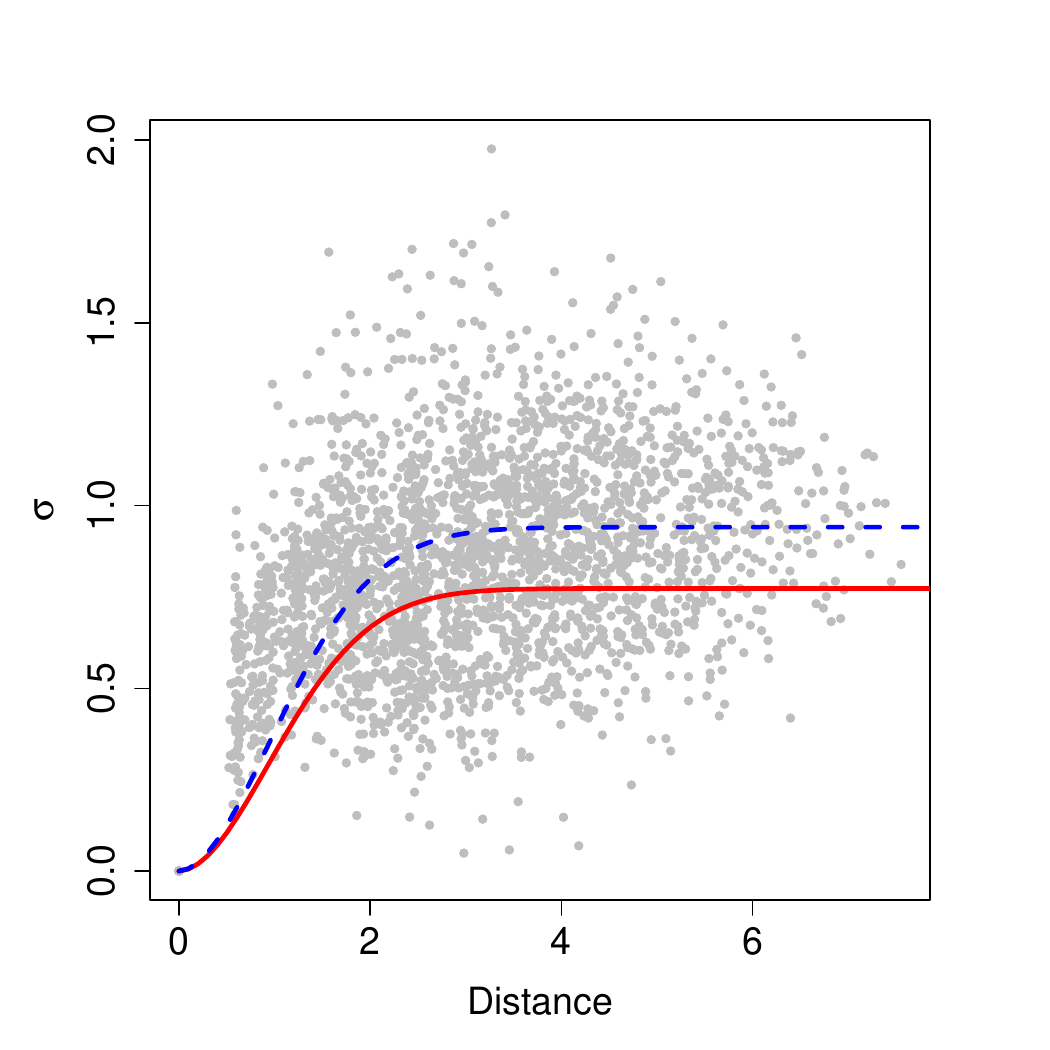}
 \caption{Parameter estimates from pairwise model fits using the same structure as Model 3, i.e., $a(x)=\alpha x$, $b(x)=1+(\alpha x)^\beta$ and $Z$ following a delta Laplace distribution. Here $\beta$ was constrained between $(-1,1)$. Fits were made using composite likelihoods where a single set of parameters was assumed to apply for each pair, whichever the conditioning variable. Distance is in the deformed coordinate space. Solid red lines display implied values from the full fitted model; dashed blue lines display implied estimates from the model refitted to the extracted residual processes $Z^j$.}
  \label{fig:pairs}
\end{figure}

\begin{figure}[H]
 \centering
 \includegraphics[width=0.4\textwidth]{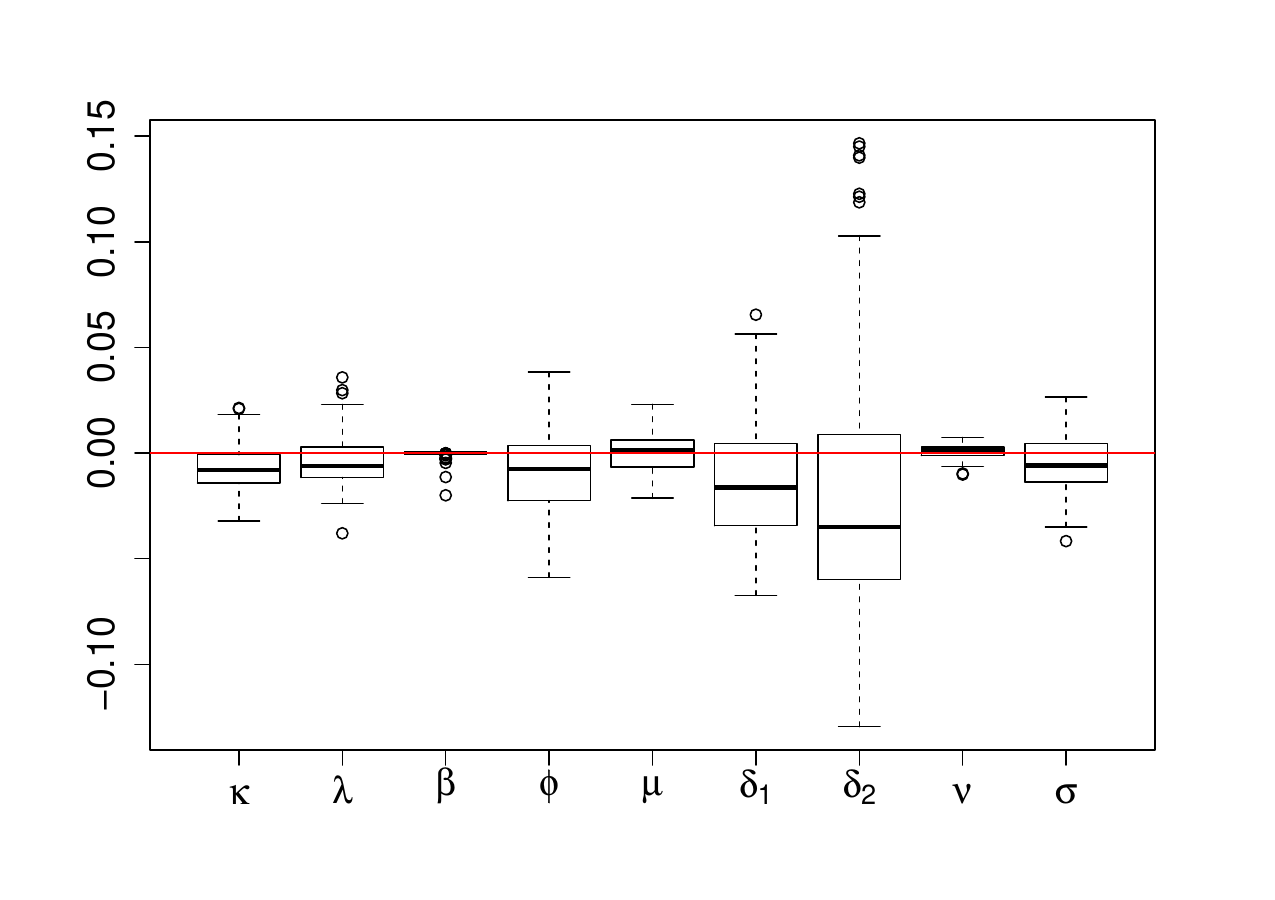}
 \includegraphics[width=0.4\textwidth]{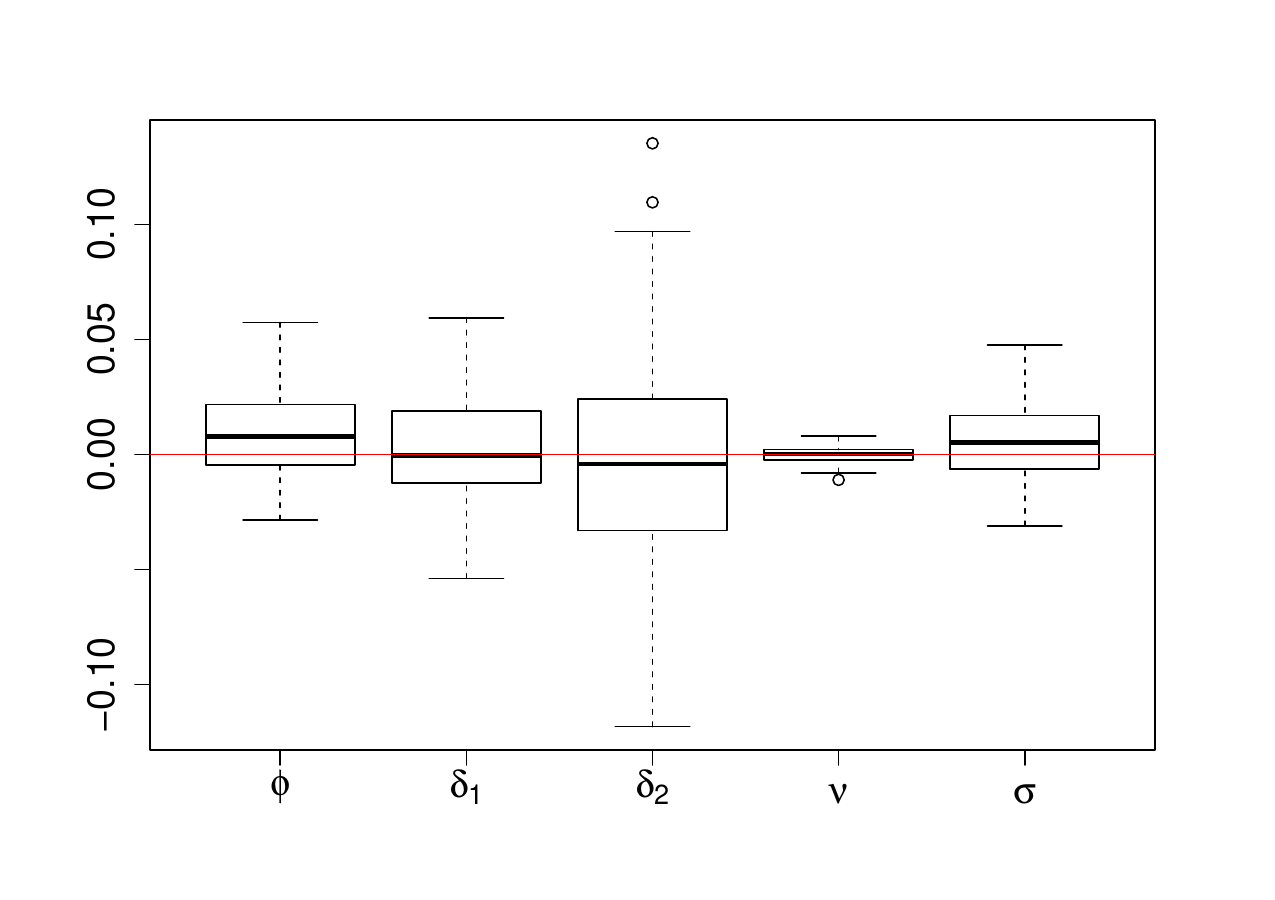}
 \caption{Distribution of estimates from 100 bootstrap repetitions, with population MLE subtracted. Left: full model fit; right: fit to extracted $Z^j$ using empirical means.}
 \label{fig:bootstrap}
\end{figure}

%--------------------------------------

\section{Computation time}
\label{app:Comptime}

\begin{figure}[h]
\centering
 \includegraphics[width=0.4\textwidth]{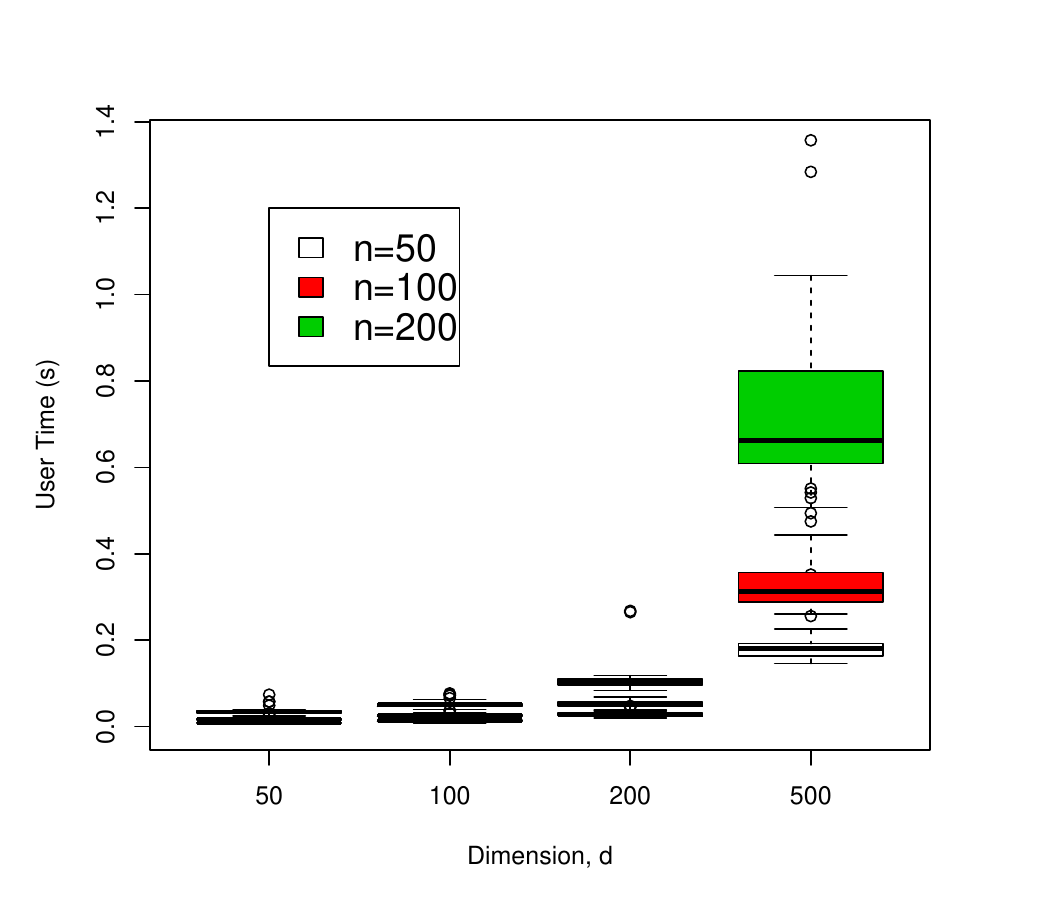}
 \includegraphics[width=0.4\textwidth]{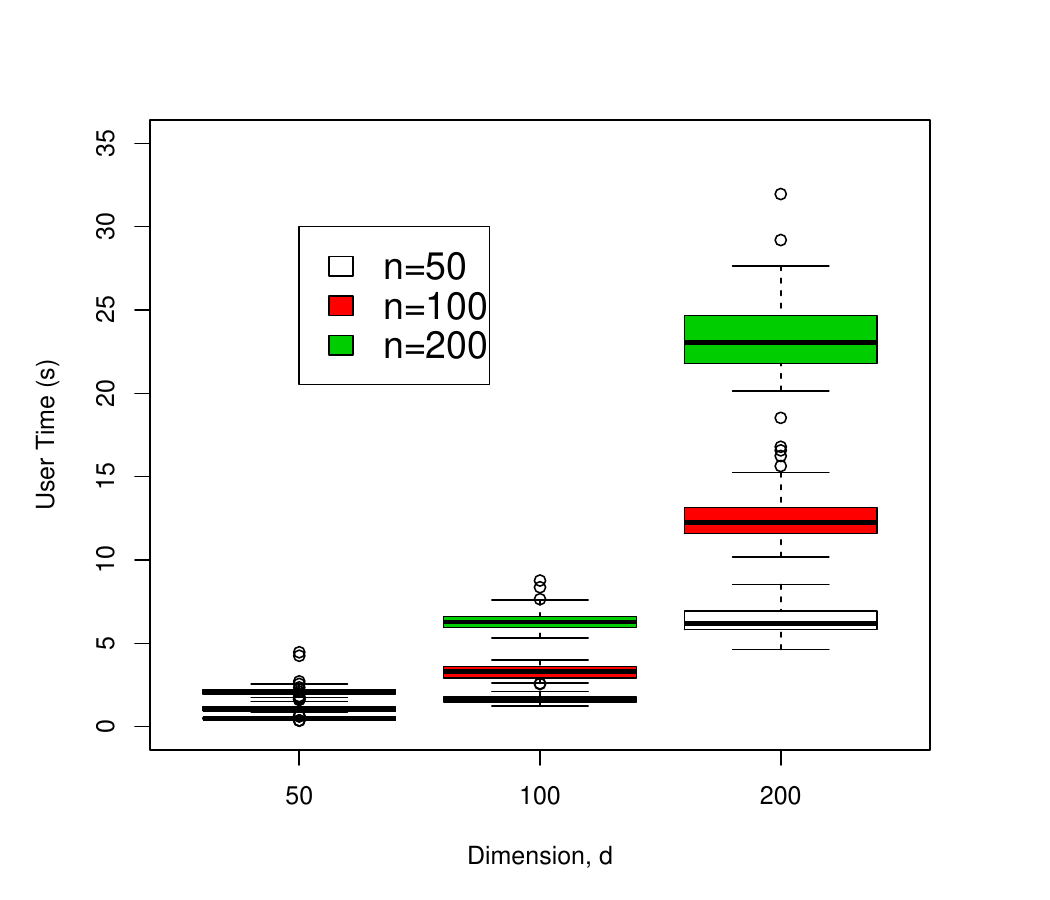}
 \caption{Left: computation time for a likelihood conditioning on a single site exceeding the threshold, for different numbers of sites, $d$, and average numbers of exceedances at the conditioning site, $n$. Right: computation time for a composite likelihood over all sites. The boxplots are created from 50 repetitions on different data.}
  \label{fig:comptime}
\end{figure}
Figure~\ref{fig:comptime} is intended as a rough guide to how many dimensions one might reasonably attempt to optimize a likelihood in, based on our R code implementation, which makes use of the \texttt{mvnfast} R library \citep{Fasiolo16}. Computation times for the composite likelihood generally slightly exceed $d$ times the computation for the likelihood conditioning at a single site.

%--------------------------------------

\bibliographystyle{apalike}
\bibliography{CESBib}
\end{document}